\documentclass{CSML}

\pdfoutput=1

\usepackage{lastpage}

\def\dOi{13(4:11)2017}
\lmcsheading%
{\dOi}
{1--\pageref{LastPage}}
{}
{}
{Mar.~\phantom04, 2016}
{Nov.~17, 2017}
{}

\usepackage[inference]{semantic}
\usepackage{verbatim}
\usepackage{stmaryrd}
\usepackage[reqno]{amsmath}
\usepackage{amsthm}
\usepackage{amssymb}
\usepackage{graphicx}
\usepackage{xypic}
\usepackage{bbm}
\usepackage{pifont}
\usepackage{txfonts}

\usepackage[hidelinks]{hyperref}

\begin{document}

\title[The Universal Process]{The Universal Process}

\author[Y.~Fu]{Yuxi Fu}
\address{BASICS, Department of Computer Science, Shanghai Jiaotong
  University, Shanghai, China}	
\email{fu-yx@cs.sjtu.edu.cn}
\thanks{Work supported by NSFC (61472239, PACE 61261130589).}

\subjclass{[{\bf Theory of computation}]: Models of computation---Concurrency}

\begin{abstract}
\noindent A universal process of a process calculus is one that, given the G\"{o}del index of a process of a certain type, produces a process equivalent to the encoded process.
This paper demonstrates how universal processes can be formally defined and how a universal process of the value-passing calculus can be constructed.
The existence of such a universal process in a process model can be explored to implement higher order communications, security protocols, and programming languages in the process model.
A process version of the S-m-n theorem is stated to showcase how to embed the recursion theory in a process calculus.
\end{abstract}

\maketitle

\section{Introduction}\label{sec-Introduction}

The classic recursion theory~\cite{Rogers1987,Soare1987} is based on two fundamental observations.
The first is that there is an effective function $\phi^{k}$ that enumerates all the $k$-ary recursive functions.
By fixing an enumeration function we can write $\phi^{k}_{i}$ for $\phi^{k}(i)$, the $i$-th $k$-ary recursive function.
The number $i$ is called the {\em G\"{o}del number}, or the {\em G\"{o}del index} of the recursive function.
The effectiveness of $\phi^{k}_{i}$ comes in both directions.
One can effectively calculate a unique number from a given recursive function.
One can also effectively recover a unique recursive function from a given number.
The S-m-n Theorem states that for all $k_{0},k_{1}$ there is a total $(k_{0}{+}1)$-ary recursive function $\textsf{s}_{k_{1}}^{k_{0}}(z,x_{1},\ldots,x_{k_{0}})$ such that
$\phi^{k_{0}+k_{1}}_{k}(i_{1},\ldots,i_{k_{0}},j_{1},\ldots,j_{k_{1}}) \simeq \phi^{k_{1}}_{\textsf{s}_{k_{1}}^{k_{0}}(k,i_{1},\ldots,i_{k_{0}})}(j_{1},\ldots,j_{k_{1}})$ for all numbers $k,i_{1},\ldots,i_{k_{0}},j_{1},\ldots,j_{k_{1}}$.
The equality $\simeq$ means that either both sides are defined and they are equal or neither side is defined.
The second important observation is that there exists a $(k{+}1)$-ary universal function $\mathcal{U}^{k}$ that, upon receiving an index $j$ of a $k$-ary recursive function $\textsf{f}$ and $k$ numbers $i_{1},\ldots,i_{k}$, evaluates $\textsf{f}(i_{1},\ldots,i_{k})$.
In other words, $\mathcal{U}^{k}(j,i_{1},\ldots,i_{k}) \simeq \phi^{k}_{j}(i_{1},\ldots,i_{k})$.
The existence of such a universal function depends crucially on G\"{o}delization.
It is by G\"{o}delization that we can see a number both as a datum and a program.
The S-m-n Theorem and the universal functions are the foundational tools in recursion theory.
The practical counterpart of a universal function is a general purpose computer.
The central idea of the von Neumann structure of such a computer is that of the stored program, which is essentially the same thing as G\"{o}delization.
From the point of view of programming, a universal function is an interpreter that works by interpreting a datum as a program.
Again this is the idea of G\"{o}delization.

Recursion theory plays a foundational role in computation theory and the theory of programming languages.
It makes one think why in the theory of process calculus, or more generally in concurrency theory, the fundamental technique of G\"{o}delization has not been utilized so far.
One possible explanation is that concurrent computations are often distributed.
For processes scattered at different locations the notion of a centralized universal process may sound alien.
In retrospect however, the absence of any universal process has been unfortunate.
The $\pi$-calculus~\cite{MilnerParrowWalker1992}, and CCS~\cite{Milner1989} as well, was proposed with the intention to be the `$\lambda$-calculus' for concurrent computation.
Yet in the theory of process calculus there still lacks a notion comparable to that of decidability.
There are now some interesting techniques that allow one to prove negative results in process calculus~\cite{BusiGabbrielliZavattaro2003,BusiGabbrielliZavattaro2004,Palamidessi2003,GiambiagiSchneiderValencia2004,FuLu2010}.
However they do not offer a method as general as the reduction method in recursion theory.
To develop a theory of solvability or definability for process calculus, the ideas and the techniques of recursion theory are instructive.
In programming theory, there have been quite a few papers on implementing variants of $\pi$, or substantial extensions of them, on current computing platforms.
But there has been little discussion on how to implement a concurrent programming language in the $\pi$-calculus.
It's understandably so since the idea of a universal process (or general interpreter) is indispensable in any such implementation.
If we are serious about the claim that the $\pi$-calculus is to concurrent computation what the $\lambda$-calculus is to functional computation, we should look at implementation issues of concurrent programming languages in the $\pi$-calculus.

The above discussion leads to the conclusion that in both theory and practice there is a genuine need for a process theory that goes beyond the classic recursion theory of function.
The theory of process calculus currently fails to meet that demand.
What can we do to improve the situation?
A natural thing to do is to look at how G\"{o}delization can be carried out in process calculi and how universal processes can be constructed.
G\"{o}delization {\em is} a problem for a process calculus that cannot even code up the natural numbers in a way that supports the interpretations of the computable functions within the calculus.
We need to confine our attention to complete models.
Intuitively a complete process calculus is one that is expressive enough to admit good use of G\"{o}delization.
Now suppose $\mathbb{M}$ is a complete model.
What does a universal process of $\mathbb{M}$ look like?
In the general case it is unlikely that there is a single $\mathbb{M}$-process capable of simulating all $\mathbb{M}$-processes.
A $\pi$-process for example only refers to a finite number of global names.
From the viewpoint of observational equality, there is no way for it to simulate a $\pi$-process that uses strictly more global names.
Our strong notion of equality completely rules out such a scenario.
A universal process of a process calculus should consist of a countable family of processes.
Luckily we seldom need a single all powerful universal process.
In most applications it suffices to have a collection of processes, each acting as a universal process for a set of processes of a certain type.
A type for example could be a finite set of names.
Then a process is of that type if the names it contains all appear in that set.
If we think of it, having to use a restricted version of universal process does not really stop us from deriving any definability/undefinability results in $\mathbb{M}$.
If something is definable in $\mathbb{M}$, it is definable by an $\mathbb{M}$-process of some type.
If it is not definable, it is not defined by any $\mathbb{M}$-process of any type.

We will look at G\"{o}delization and the notion of universal process in $\mathbb{VPC}$, a self-contained version of the value-passing calculus.
The reason to start with this particular model is that it is closer to recursion theory than all the other process calculi~\cite{FuYuxi-2013-VPC}.
The contribution of this paper is the introduction of a formal definition of universal process and the construction of a universal process for $\mathbb{VPC}$.
The significance of the existence of a universal process is emphasized by illustrating a number of applications.
The technique developed in this work is expected to play a key role in the study of process theory and programming theory implemented on process models.

The paper is structured as follows.
Section~\ref{sec-Preliminary} reviews the necessary background on $\mathbb{VPC}$ and the observational theory of processes.
Section~\ref{sec-Universal-Process} provides the formal definition of universal process and demonstrates how to construct a universal process in $\mathbb{VPC}$.
Section~\ref{sec-Application} outlines three major applications of universal process.
Section~\ref{sec-S-m-n-Theorem} formalizes the process version of S-m-n Theorem.
Section~\ref{sec-Future-Work} discusses some future research topics.

Before engaging in the technicalities in the rest of the paper, we should comment on the presentation style of this paper.
We shall not spell out all technical details of our constructions, and consequently nor shall we formally establish the correctness of the constructions.
We will make full use of the fact that $\mathbb{VPC}$ is complete.
This is very much like what recursion theoreticians make use of Church-Turing Thesis since the publication of Post's pioneering paper~\cite{Post1944}.
If one has not built up enough confidence in exploiting the completeness of process models this way, one is advised to consult~\cite{FuLu2010,FuYuxi-2013-VPC,FuYuxi-Nondetrministic-Computation,FuYuxi} in which sufficient technical details can be found.

\section{Preliminary}\label{sec-Preliminary}

In this section we define the semantics of the value-passing calculus, fix the notion of process equality used in this paper, and explain in what sense the value-passing calculus is complete.

\subsection{VPC}\label{VPC}

Value-passing calculi~\cite{Hoare1985,Milner1989,HennessyIng1993a,HennessyIng1993b,HennessyLin1995} have been studied in various contexts.
In most of these studies, the value domains are left open-ended.
A recent work that provides a self-contained account of the value-passing calculi is~\cite{FuYuxi-2013-VPC}.
Since our value-passing calculus is going to be the source model whose programs are to be interpreted by a universal process, an open-ended attitude is inadequate.
At the same time we hope to avoid the formality of~\cite{FuYuxi-2013-VPC} for clarity.
Fortunately there is a standard theory we can refer to.
The value domain of our value-passing calculus is taken to be Presburger Arithmetic~\cite{Presburger1929} (an English translation of the original paper can be found in~\cite{Stansifer1984}).
This is the sub-theory of Peano Arithmetic defined by the constant $0$, the unary function $\textsf{s}$ and the binary function `$+$'.
By overloading notations, we shall abbreviate $\textsf{s}^k(0)$ to $k$ and $\textsf{s}^k(x)$ to $x+k$.
For our purpose the most attractive property of Presburger Arithmetic is the decidability of its first order theory.
There is a terminating procedure that decides the validity of every first order formula of Presburger Arithmetic~\cite{Presburger1929,Monk1976,Enderton2001}.
This is a crucial property if a value-passing calculus is seen as a programming model.
The absence of the multiplication operator does not affect the power of our model since the operator can be implemented in the value-passing calculus~\cite{FuYuxi-2013-VPC}.

Let $\textsf{N}$ be the set $\{0,\textsf{s}(0),\textsf{s}^2(0),\ldots\}$, ranged over by $i,j,k$, and $\textsf{V}$ be the set of natural number variables, ranged over by $x,y,z$.
The set $\textsf{T}$ of {\em value terms}, ranged over by $s,t$, is constructed from the numbers, the variables, and the binary operator `$+$'.
The notation $\textsf{T}^{0}$ stands for the set of closed terms.
The set $\textsf{B}$ of first order {\em logical formulae}, ranged over by $\varphi$, consists of the formulas constructed from the terms, the logical operators $\bot,\top,\wedge,\vee,\Rightarrow,\exists,\forall$ and the binary relations $<,=$.
We write $\vdash\varphi$ if $\varphi$ is a theorem of Presburger Arithmetic.

Let $\mathcal{N}$ be the set of names, ranged over by $a,b,c,d,e,f,g,h$.
The set of the finite $\mathbb{VPC}$-terms is defined by the following BNF:
\begin{eqnarray*}
T &:=& {\bf 0} \mid a(x).T \mid \overline{a}(t).T \mid T\,|\,T \mid (c)T \mid \textit{if}\;\varphi\;\textit{then}\;T .
\end{eqnarray*}
The $\mathbb{VPC}$-processes, denoted by $P,Q$, are the $\mathbb{VPC}$-terms that contain no free variables.
The name $c$ in $(c)T$ is a local name.
A name is global if it is not local.
The semantics of the finite $\mathbb{VPC}$-terms is given by the following labeled transition system, where $\alpha$ ranges over the action set $\{a(i),\overline{a}(i) \mid a\in\mathcal{N},i\in\textsf{N}\}\cup\{\tau\}$.

\vspace*{2mm}

\noindent{\em Action}
\[\begin{array}{cc}
\inference{}{a(x).T\stackrel{a(i)}{\longrightarrow}T\{i/x\}}
\ \ \ & \inference{}{\overline{a}(t).T\stackrel{\overline{a}(i)}{\longrightarrow}T}\ \vdash t=i.
\end{array}\]
\noindent{\em Composition}
\[\begin{array}{cc}
\inference{S\stackrel{\alpha}{\longrightarrow}S'}{S\,|\,T\stackrel{\alpha}{\longrightarrow}S'\,|\,T}
\ \ \  &
\inference{S\stackrel{a(i)}{\longrightarrow}S'\ \ \ \
\
T\stackrel{\overline{a}(i)}{\longrightarrow}T'}{S\,|\,T\stackrel{\tau}{\longrightarrow}S'\,|\,T'}
\end{array}\]
\noindent{\em Localization}
\[\inference{T\stackrel{\alpha}{\longrightarrow}T'}
{(c)T\stackrel{\alpha}{\longrightarrow}(c)T'} \ c\ \mathrm{is}\ \mathrm{not}\ \mathrm{in}\ \alpha.\]
\noindent{\em Condition}
\[\inference{T\stackrel{\alpha}{\longrightarrow}T'}
{\textit{if}\;\varphi\;\textit{then}\;T\stackrel{\alpha}{\longrightarrow}T'} \ \vdash\varphi.\]

\vspace*{1mm}

\noindent
We shall use standard notations like $\Longrightarrow$ and $\stackrel{\alpha}{\Longrightarrow}$.
The recursion mechanism of a value-passing calculus can be defined in a number of ways.
They are not completely equivalent in terms of expressive power~\cite{BusiGabbrielliZavattaro2003,BusiGabbrielliZavattaro2004,Palamidessi2003,GiambiagiSchneiderValencia2004,FuLu2010}.
The infinite behaviors of our model $\mathbb{VPC}$ is introduced by equationally defined terms.
A {\em parametric definition} is given by the equation
\begin{equation}\label{2011-08-21}
D(x_{1},\ldots,x_{k}) = T,
\end{equation}
where $x_{1},\ldots,x_{k}$ are parameter variables.
In this paper we require that $T$ does not contain any free variable not in $\{x_{1},\ldots,x_{k}\}$.
The instantiation of $D(x_{1},\ldots,x_{k})$ at value terms $t_{1},\ldots,t_{k}$, denoted by $D(t_{1},\ldots,t_{k})$, is $T\{t_{1}/x_{1},\ldots,t_{k}/x_{k}\}$.
In addition to the finite terms, $\mathbb{VPC}$ also has instantiated terms of the form $D(t_{1},\ldots,t_{k})$, where $t_{1},\ldots,t_{k}$ are value terms.
The parametric definition (\ref{2011-08-21}) is generally recursive in the sense that $T$ may contain instantiated occurrences of $D(x_{1},\ldots,x_{k})$.
It may also contain instantiated occurrences of some $D'(y_{1},\ldots,y_{j})$ given by another parametric definition.
The operational semantics of $D(t_{1},\ldots,t_{k})$ is defined by the following rule:
\[\inference{T\{t_{1}/x_{1},\ldots,t_{k}/x_{k}\}\stackrel{\alpha}{\longrightarrow}T'}{D(t_{1},\ldots,t_{k})\stackrel{\alpha}{\longrightarrow}T'}\ D(x_{1},\ldots,x_{k}) = T.\]
Now suppose $D(x)=\overline{c}(0)\,|\,(c)(\overline{c}(x)\,|\,\overline{c}(x)\,|\,c(z).D(z{+}1))$.
Then the following reductions are admissible:
\begin{eqnarray*}
D(1) &\stackrel{\tau}{\longrightarrow}& \overline{c}(0)\,|\,(c)(\overline{c}(1)\,|\,D(1+1)) \\
 &\stackrel{\tau}{\longrightarrow}& \overline{c}(0)\,|\,(c)(\overline{c}(1)\,|\,\overline{c}(0)\,|\,(c)(\overline{c}(2)\,|\,D(2{+}1))).
\end{eqnarray*}
In this example the global name $c$ in the component $\overline{c}(0)$ gets captured every time the parametric definition is unfolded.

An alternative to parametric definition is replication.
The syntax for the replication terms is given by
\begin{eqnarray*}
T &:=& \ldots \mid \;!a(x).T \mid \;!\overline{a}(t).T.
\end{eqnarray*}
The operational semantics of the replicator ``$!$'' is defined by the following transitions:
\[\begin{array}{cc}
\inference{}{!a(x).T\stackrel{a(i)}{\longrightarrow}T\{i/x\}\,|\,!a(x).T}
\ \ \ & \inference{}{!\overline{a}(t).T\stackrel{\overline{a}(i)}{\longrightarrow}T\,|\,!\overline{a}(t).T}\ \vdash t=i.
\end{array}\]
We will denote by $\mathbb{VPC}^{!}$ the value-passing calculus with the replicator.

The replicator is a derived operator in $\mathbb{VPC}$.
The terms $!a(x).T$ and $!\overline{a}(t).T$ are equal to the instantiations of the following abstractions respectively.
\begin{eqnarray}
C(x_{1},\ldots,x_{k}) &=& a(x).S\,|\,C(x_{1},\ldots,x_{k}), \label{2016-04-03-a} \\
D(x_{1},\ldots,x_{k}) &=& \overline{a}(t).T\,|\,D(x_{1},\ldots,x_{k}), \label{2016-04-03-b}
\end{eqnarray}
where $\{x_{1},\ldots,x_{k}\}$ is the set of the free variables appearing in $!a(x).T$ respectively $!\overline{a}(t).T$.
We shall freely use the replication operator in $\mathbb{VPC}$.
Leaving aside the question if $\mathbb{VPC}^{!}$ is as expressive as $\mathbb{VPC}$ for the moment, we point out that all recursive functions can be implemented in $\mathbb{VPC}^{!}$~\cite{FuYuxi-2013-VPC}.

The following abbreviations will be used
\begin{eqnarray*}
a.T &\stackrel{\rm def}{=}& a(x).T,\ \ \mathrm{where}\ x\ \mathrm{does}\ \mathrm{not}\ \mathrm{appear}\ \mathrm{in}\ T, \\
\overline{a}.T &\stackrel{\rm def}{=}& \overline{a}(0).T.
\end{eqnarray*}
We occasionally write for example $t(x)$ to indicate that $t$ contains the variable $x$.
Accordingly we write $t(s)$ for the term obtained by substituting $s$ for $x$.
The notations $\varphi(x),\varphi(s)$ and $T(x),T(s)$ are used similarly.
We sometimes use the two leg if command defined as follows:
\begin{eqnarray*}
\textit{if}\;\varphi\;\textit{then}\;S\;\textit{else}\;T &\stackrel{\rm def}{=}& \textit{if}\;\varphi\;\textit{then}\;S \,|\, \textit{if}\;\neg\varphi\;\textit{then}\;T.
\end{eqnarray*}
For clarity we will write
\begin{eqnarray*}
&& {\bf case}\ t\ \mathbf{of} \\
&& \ \ \varphi_{0}(z)\Rightarrow T_{0}(z); \\
&& \ \ \ \ \ \ \ \vdots \\
&& \ \ \varphi_{k-1}(z)\Rightarrow T_{k-1}(z); \\
&& \ \ \varphi_{k}(z)\Rightarrow T_{k}(z) \\
&& {\bf end}\ {\bf case}
\end{eqnarray*}
for the nested if statement $\textit{if}\;\varphi_{0}(t)\;\textit{then}\;T_{0}(t)\;\textit{else}\;\textit{if}\;\ldots\;\textit{else}\;\textit{if}\;\varphi_{k}(t)\;\textit{then}\;T_{k}(t)$.
The auxiliary notation $\textit{let}\;x=t\;\textit{in}\;T$ stands for $T\{t/x\}$.
This is useful when $t$ is a long expression and $x$  occurs in $T$ multiple times.

For a complete treatment of $\mathbb{VPC}$ the reader should consult~\cite{FuYuxi-2013-VPC}.

\subsection{Equality and Expressiveness}\label{Absolute-Equality-and-Expressiveness}

The definition of a universal process must refer to a process equality.
The choice of such an equality is not entirely orthogonal to the existence of a universal process.
It is conceivable that some sort of universal process exists with respect to a weak equality, whereas it is impossible to have a universal process with respect to a stronger equality.
To present our result in its strongest form, we shall introduce a number of properties that we believe best describe the {\em correctness} of our universal processes.
The following account follows the general methodology of~\cite{FuYuxi}.
The description given here is however self-contained.
In this section we assume that $\mathbb{M}$ is a process calculus and $\mathcal{R}$ is a binary relation on the set of $\mathbb{M}$-processes.
The notation $\mathcal{R}^{-1}$ stands for the reverse relation of $\mathcal{R}$.

A universal process is a generalization of a universal Turing machine.
Upon receiving a number the latter simulates the Turing machine encoded by the number.
But how about the correctness of the simulation?
The answer is provided by the operational interpretation of the Church-Turing Thesis~\cite{vanEmdeBoas1990}.
A sound translation of one computation model to another is a bisimulation of computation steps {\em \`{a} la} Milner~\cite{Milner1989} and Park~\cite{Park1981}.
Moreover if we take nondeterministic computation into account the translation ought to be a branching bisimulation of van Glabbeek and Weijland~\cite{vanGlabbeekWeijland1989-first-paper-bb}.
The reader is referred to~\cite{FuYuxi-Nondetrministic-Computation} for a formal study of nondeterministic computation in a process algebraic setting.
\begin{defi}
$\mathcal{R}$ is a {\em bisimulation} if the following clauses are valid whenever $P\mathcal{R}Q$:
\begin{enumerate}
\item If $P\stackrel{\tau}{\longrightarrow}P'$ then one of the
following statements is valid:
\begin{enumerate}[label=(\roman*)]
\item $Q\Longrightarrow Q'$ for some $Q'$ such that $Q'\mathcal{R}^{-1}P$ and $Q'\mathcal{R}^{-1}P'$.

\item $Q\Longrightarrow
Q''\mathcal{R}^{-1}P$ for some $Q''$ such that $Q''\stackrel{\tau}{\longrightarrow}Q'\mathcal{R}^{-1}P'$ for some $Q'$.
\end{enumerate}

\item If $Q\stackrel{\tau}{\longrightarrow}Q'$ then one of the
following statements is valid:
\begin{enumerate}[label=(\roman*)]
\item $P\Longrightarrow P'$ for some $P'$ such that $P'\mathcal{R}Q$ and $P'\mathcal{R}Q'$.

\item $P\Longrightarrow
P''\mathcal{R}Q$ for some $P''$ such that $P''\stackrel{\tau}{\longrightarrow}P'\mathcal{R}Q'$ for some $P'$.
\end{enumerate}
\end{enumerate}
\end{defi}

\noindent A universal process must be sensitive to divergence.
It would be unacceptable to interpret all processes by divergent processes.
The following definition is from~\cite{Priese1978}.
It is the best formalization of the termination preserving property that goes along with the bisimulations~\cite{vanGlabbeekLuttikTrcka2009,FuYuxi}.
\begin{defi}
$\mathcal{R}$ is {\em codivergent} if the following statements are valid whenever $P\mathcal{R}Q$:
\begin{enumerate}
\item
If there is an infinite internal transition sequence $Q\stackrel{\tau}{\longrightarrow}Q_{1}\stackrel{\tau}{\longrightarrow}\ldots\stackrel{\tau}{\longrightarrow}Q_{i}\stackrel{\tau}{\longrightarrow}\ldots$ then $P\stackrel{\tau}{\Longrightarrow}P'\mathcal{R}Q_{k}$ for some $P'$ and some $k\ge1$.
\item
If there is an infinite internal transition sequence $P\stackrel{\tau}{\longrightarrow}P_{1}\stackrel{\tau}{\longrightarrow}\ldots\stackrel{\tau}{\longrightarrow}P_{i}\stackrel{\tau}{\longrightarrow}\ldots$ then $Q\stackrel{\tau}{\Longrightarrow}Q'\mathcal{R}^{-1}P_{k}$ for some $Q'$ and some $k\ge1$.
\end{enumerate}
\end{defi}

\noindent A universal process should also respect interactability.
The barbedness of Milner and Sangiorgi~\cite{MilnerSangiorgi1992} poses a minimal condition.
We say that a process $P$ is observable, notation $P{\Downarrow}$, if $\Longrightarrow\stackrel{\alpha}{\longrightarrow}$ for some $\alpha\ne\tau$.
It is unobservable, notation $P{\not\Downarrow}$, if it is not observable.
\begin{defi}
$\mathcal{R}$ is {\em equipollent} if $P{\Downarrow}\Leftrightarrow Q{\Downarrow}$ whenever $P\mathcal{R}Q$.
\end{defi}
The equipollence condition does not make much sense unless some kind of closure property is available.
Concurrent composition and localization are the most fundamental operators in concurrency theory.
The former makes global interaction possible while the latter localizes such possibility.
Hence the following definition.
\begin{defi}
$\mathcal{R}$ is {\em extensional} if the following statements are valid:
\begin{enumerate}
\item If $P\mathcal{R}Q$ then $(c)P\mathcal{R}(c)Q$ for every $c\in\mathcal{N}$;
\item If $P_{0}\mathcal{R}Q_{0}$ and $P_{1}\mathcal{R}Q_{1}$ then $(P_{0}\,|\,P_{1})\mathcal{R}(Q_{0}\,|\,Q_{1})$.
\end{enumerate}
\end{defi}

In~\cite{FuYuxi} it is argued that these properties give a model-independent characterization of process equality.
\begin{defi}\label{absolute-equality}
The absolute equality $=_{\mathbb{M}}$ is the largest relation on the set of the $\mathbb{M}$-processes that satisfies the following:
\begin{enumerate}
\item It is reflexive;

\item \label{2016-03-31} It is extensional, equipollent, codivergent and bisimilar.
\end{enumerate}
\end{defi}
It is easy to convince oneself that $=_{\mathbb{M}}$ is well defined.
So we have $=_{\mathbb{VPC}}$ and $=_{\mathbb{VPC}^{!}}$.
We will often omit the subscript.

The abstract definition of $=_{\mathbb{M}}$ makes it difficult to work with.
We need a characterization of $=_{\mathbb{M}}$ that relies neither on the equipollence condition nor on the extensionality condition.
In practice it is sufficient to have an external bisimilarity $\simeq_{\mathbb{M}}$ satisfying $\simeq_{\mathbb{M}}\;\subseteq\;=_{\mathbb{M}}$.
\begin{defi}\label{external-bi}
A codivergent bisimulation $\mathcal{R}$ is an $\mathbb{M}$-bisimulation if the following statements are valid whenever $P\mathcal{R}Q$ and $\alpha\ne\tau$:
\begin{enumerate}
\item If $P\stackrel{\alpha}{\longrightarrow}P'$ then
$Q\Longrightarrow Q''\stackrel{\alpha}{\longrightarrow}Q'\mathcal{R}^{-1}P'$ and $P\mathcal{R}Q''$ for some $Q',Q''$.
\item If $Q\stackrel{\alpha}{\longrightarrow}Q'$ then
$P\Longrightarrow P''\stackrel{\alpha}{\longrightarrow}P'\mathcal{R}Q'$ and $P''\mathcal{R}Q$ for some $P',P''$.
\end{enumerate}
The $\mathbb{M}$-bisimilarity $\simeq_{\mathbb{M}}$ is the largest $\mathbb{M}$-bisimulation.
\end{defi}
Both $\simeq_{\mathbb{VPC}}\;\subseteq\;=_{\mathbb{VPC}}$ and $\simeq_{\mathbb{VPC}^{!}}\;\subseteq\;=_{\mathbb{VPC}^{!}}$ hold.
We shall make use of these facts in the correctness proofs.

By making use of the congruence $=$ we can define semantically the one step {\em deterministic} computation $P\rightarrow P'$ as an internal action $P\stackrel{\tau}{\longrightarrow}P'$ such that $P'=P$, and the one step {\em nondeterministic} computation $P\stackrel{\iota}{\longrightarrow}P'$ as an internal action $P\stackrel{\tau}{\longrightarrow}P'$ such that $P'\ne P$. The rich structure of nondeterministic computation is discussed in~\cite{FuYuxi-Nondetrministic-Computation}.
The distinction between the two classes of internal actions is important to appreciate the working mechanism of the universal process.
The reflexive and transitive closure of $\rightarrow$ is denoted by $\rightarrow^{*}$.

If we consider interpreter rather than universal process, we need to relate a process of one model to a process of another model.
In other words we need to talk about `equality' between the processes from two different process calculi.
This way of looking at the expressiveness relationship between two process calculi leads immediately to the notion of subbisimilarity~\cite{FuYuxi}.
The reflexivity of Definition~\ref{absolute-equality} is turned into totality and soundness.
Totality means that for each $\mathbb{M}$-process $P$ there is an $\mathbb{N}$-process $Q$ such that $P\propto Q$.
Soundness is the condition that $Q=_{\mathbb{N}}Q'$ whenever $P=_{\mathbb{M}}P'$, $P\propto Q$ and $P'\propto Q'$.
\begin{defi}\label{subbisimilarity}
A binary relation $\propto$ from the set of $\mathbb{M}$-processes to the set of $\mathbb{N}$-processes is a {\em subbisimilarity} if it renders true the following statements:
\begin{enumerate}
\item $\propto$ is total and sound;
\item \label{2016-04-01} $\propto$ is extensional, equipollent, codivergent and bisimilar.
\end{enumerate}
$\mathbb{M}$ is {\em subbisimilar to} $\mathbb{N}$, notation $\mathbb{M}\sqsubseteq\mathbb{N}$ or $\mathbb{N}\sqsupseteq\mathbb{M}$, if there is a subbisimilarity, a witness of $\mathbb{M}\sqsubseteq\mathbb{N}$, from $\mathbb{M}$ to $\mathbb{N}$.
\end{defi}

Intuitively $\mathbb{M}\sqsubseteq\mathbb{N}$ means that $\mathbb{N}$ is at least as expressive as $\mathbb{M}$.
The existence of a total relation $\propto$ from $\mathbb{M}$ to $\mathbb{N}$ means that every $\mathbb{M}$-process $P$ gets interpreted by some $\mathbb{N}$-process $Q$.
The condition (\ref{2016-04-01}) of Definition~\ref{subbisimilarity}, which is the same as the condition (\ref{2016-03-31}) of Definition~\ref{absolute-equality}, guarantees that $P$ and $Q$ behave the same in respective models.
Furthermore the soundness allows one to see $P$ and $Q$ as equal as it were.
The relation $\sqsubseteq$ is stronger than most of the expressiveness relations discussed in the literature~\cite{BusiGabbrielliZavattaro2003,BusiGabbrielliZavattaro2004,Gorla2008Concur,Palamidessi2003,FuYuxi}, which makes the correctness of our interpreter more convincing.

\subsection{Completeness}\label{sec-Completeness}

Both $\mathbb{VPC}$ and $\mathbb{VPC}^{!}$ are Turing complete.
There are several interpretations of Turing completeness in the literature on process calculus~\cite{BusiGabbrielliZavattaro2004,MaffeisPhillips2005,FuLu2010}.
The general requirement on Turing completeness of a process model $\mathbb{M}$ can be summarized as follows:
\begin{itemize}
\item There is an encoding $\llbracket\_\rrbracket$ of the natural numbers in $\mathbb{M}$.

\item There is an interpretation $\llbracket\_\rrbracket$ of recursive functions~\cite{Rogers1987} in $\mathbb{M}$ such that for every $k$-ary computable function $\textsf{f}(x_{1},\ldots,x_{k})$ and all numbers $i_{1},\ldots,i_{k}$ the following operational property holds:
    If $\textsf{f}(i_{1},\ldots,i_{k})$ is defined then
    \[\llbracket i_{1}\rrbracket\,|\,\ldots\,|\,\llbracket i_{k}\rrbracket\,|\,\llbracket \textsf{f}(x_{1},\ldots,x_{k})\rrbracket\stackrel{\tau}{\Longrightarrow}\approx\llbracket \textsf{f}(i_{1},\ldots,i_{k})\rrbracket,\]
\end{itemize}
where $\stackrel{\tau}{\Longrightarrow}$ is the transitive closure of $\stackrel{\tau}{\longrightarrow}$;
if $\textsf{f}(i_{1},\ldots,i_{k})$ is undefined then
    \[\llbracket i_{1}\rrbracket\,|\,\ldots\,|\,\llbracket i_{k}\rrbracket\,|\,\llbracket \textsf{f}(x_{1},\ldots,x_{k})\rrbracket\stackrel{\tau}{\Longrightarrow}\approx\Omega,\]
where $\Omega$ is a divergent process whose only action is $\Omega\stackrel{\tau}{\longrightarrow}\Omega$, and $\approx$ is one of the termination preserving weak equalities.
A criticism to this level of completeness is that the input numbers $\llbracket i_{1}\rrbracket,\ldots,\llbracket i_{k}\rrbracket$ are not necessarily picked up properly by $\llbracket \textsf{f}(x_{1},\ldots,x_{k})\rrbracket$, and the result number $\llbracket \textsf{f}(x_{1},\ldots,x_{k})\rrbracket$ is not sent to any intended target.
The evolution from $\llbracket i_{1}\rrbracket\,|\,\ldots\,|\,\llbracket i_{k}\rrbracket\,|\,\llbracket \textsf{f}(x_{1},\ldots,x_{k})\rrbracket$ to $\llbracket \textsf{f}(i_{1},\ldots,i_{k})\rrbracket$ could be too liberal.

The Turing completeness of an interaction model $\mathbb{M}$ means that an {\em outsider} can see that the recursive functions can be coded up using $\mathbb{M}$-processes.
It is an external completeness.
A stronger notion of completeness, a much more useful one in practice, is internal completeness.
Intuitively the internal completeness of $\mathbb{M}$ means that the {\em insiders}, the $\mathbb{M}$-processes, are aware of the fact that they can compute all the computable functions.
A formal treatment of this kind of completeness is provided in~\cite{FuYuxi}.
In this paper it suffices to say that the completeness of $\mathbb{M}$ boils down to the following:
\begin{itemize}
\item For each name $a$ and each number $i$ there is an encoding $\llbracket i\rrbracket_{a}$ of $i$ at $a$.

\item For all $k\ge0$ and all names $a_{1},\ldots,a_{k},b$, there is an encoding function $\llbracket\_\rrbracket_{a_{1},\ldots,a_{k}}^{b}$ such that
for every $k$-ary recursive function $\textsf{f}(x_{1},\ldots,x_{k})$ the following statement is valid:
For all natural numbers $i_{1},\ldots, i_{k}$, $\textsf{f}( i_{1},\ldots, i_{k})$ is defined if and only if
\[\llbracket i_{1}\rrbracket_{a_{1}}\,|\,\ldots\,|\,\llbracket i_{k}\rrbracket_{a_{k}}\,|\,\llbracket\textsf{f}(x_{1},\ldots,x_{k})\rrbracket_{a_{1},\ldots,a_{k}}^{b}
\;\underset{k\ \mathrm{times}}{\underbrace{\stackrel{\iota}{\longrightarrow}\ldots\stackrel{\iota}{\longrightarrow}}} =_{\mathbb{M}} \llbracket\textsf{f}( i_{1},\ldots, i_{k})\rrbracket_{b},\]
and
$\textsf{f}( i_{1},\ldots, i_{k})$ is undefined if and only if
\[\llbracket i_{1}\rrbracket_{a_{1}}\,|\,\ldots\,|\,\llbracket i_{k}\rrbracket_{a_{k}}\,|\,\llbracket\textsf{f}(x_{1},\ldots,x_{k})\rrbracket_{a_{1},\ldots,a_{k}}^{b}
\;\underset{k\ \mathrm{times}}{\underbrace{\stackrel{\iota}{\longrightarrow}\ldots\stackrel{\iota}{\longrightarrow}}} =_{\mathbb{M}} \Omega.\]
\end{itemize}
The class $\{\llbracket i\rrbracket_{a}\}_{i\in\textsf{N},a\in\mathcal{N}}$ provides an encoding of the natural numbers in $\mathbb{M}$.
The process $\llbracket i\rrbracket_{a}$ is ready to deliver the number $i$ to a process at channel $a$.
The process $\llbracket\textsf{f}(x_{1},\ldots,x_{k})\rrbracket_{a_{1},\ldots,a_{k}}^{b}$ inputs the numbers $i_{1},\ldots,i_{k}$ sequentially at $a_{1},\ldots,a_{k}$ in $k$ steps, after which it becomes some process, say $M$, equal to $\llbracket\textsf{f}( i_{1},\ldots, i_{k})\rrbracket_{b}$.
The only action of $\llbracket\textsf{f}(i_{1},\ldots, i_{k})\rrbracket_{b}$ is to deliver the result to whichever process wants a number at channel $b$.
It should be remarked that $M$ may perform a finite sequence of deterministic computations to simulate the computation of $\textsf{f}( i_{1},\ldots, i_{k})$.
All the intermediate states of this simulation are equal to each other.

Both $\mathbb{VPC}$ and $\mathbb{VPC}^{!}$ are complete in the above sense~\cite{FuYuxi-2013-VPC,FuYuxi}.
In both models the encoding $\llbracket i\rrbracket_{a}$ is defined by $\overline{a}(i)$.

The practical implication of the completeness of $\mathbb{VPC}$ and $\mathbb{VPC}^{!}$ is that we may make use of a process without explicitly defining it.
Let's explain this point by examples.
Suppose $\textsf{f}(x_{1},\ldots,x_{k_{0}})$, $\textsf{g}(y_{1},\ldots,y_{k_{1}})$ are computable functions.
Then we may assume that
\[\textit{if}\;\textsf{f}(x_{1},\ldots,x_{k_{0}})=\textsf{g}(y_{1},\ldots,y_{k_{1}})\;\textit{then}\;T\]
is a $\mathbb{VPC}$-term with the obvious semantics.
In fact it can be defined by the following term
\[(c)(d)(\llbracket\textsf{f}(x_{1},\ldots,x_{k_{0}})\rrbracket^{c}\,|\,\llbracket\textsf{g}(y_{1},\ldots,y_{k_{1}})\rrbracket^{d}\,|\,c(x).d(y).\textit{if}\;x=y\;\textit{then}\;T),\]
where $\llbracket\textsf{f}(x_{1},\ldots,x_{k_{0}})\rrbracket^{c}$ is the $\mathbb{VPC}$-process that outputs the value of $\textsf{f}(x_{1},\ldots,x_{k_{0}})$ at $c$, whose existence is guaranteed by the completeness of $\mathbb{VPC}$.
The process $\llbracket\textsf{g}(y_{1},\ldots,y_{k_{1}})\rrbracket^{d}$ is similar.
In the same fashion we may think of $\overline{a}(\textsf{f}(x_{1},\ldots,x_{k})).T$ as the $\mathbb{VPC}$-term
$(c)(\llbracket\textsf{f}(x_{1},\ldots,x_{k})\rrbracket^{c}\,|\,c(z).\overline{a}(z).T)$.
More generally let $\psi(x_{1},\ldots,x_{k})$ be a semi-decidable property and $\chi_{\psi}(x_{1},\ldots,x_{k})$ be the partial characteristic function of $\psi$.
According to definition $\chi_{\psi}(x_{1},\ldots,x_{k})$ is the recursive function that returns `$1$' at input sequence $i_{1},\ldots,i_{k}$ when the property holds and diverges otherwise.
Now we may regard $\textit{if}\;\psi\;\textit{then}\;T$ the same as $\textit{if}\;\chi_{\psi}=1\;\textit{then}\;T$.
If $\psi$ is a decidable property, then $\textit{if}\;\psi\;\textit{then}\;T\;\textit{else}\;T'$ can be interpreted as \[\textit{if}\;\chi_{\psi}=1\;\textit{then}\;T\,|\,\textit{if}\;\chi_{\psi}\ne1\;\textit{then}\;T'.\]
To simplify notation we shall use more liberal terms like
\begin{equation}\label{2011-09-02}
A(j).\textit{if}\;j=1\;\textit{then}\;T(j).
\end{equation}
In the above term the generalized prefix operation $A(j)$ is understood as an arithmetic operation.
After the result $j$ has been calculated, the process $\textit{if}\;j=1\;\textit{then}\;T(j)$ is ready to fire.
By the completeness the process in (\ref{2011-09-02}) can be implemented in $\mathbb{VPC}$.

In the rest of the paper we shall use the internal completeness of $\mathbb{VPC}$ extensively in the manner just described.

For a systematic development of the equality theory, the expressiveness theory and the completeness theory from which the definitions given in Section~\ref{Absolute-Equality-and-Expressiveness} and Section~\ref{sec-Completeness} are imported, the reader is referred to~\cite{FuYuxi}.
The completeness of $\mathbb{VPC}$ is formally established in~\cite{FuYuxi-2013-VPC}.

\section{Universal Process}\label{sec-Universal-Process}

This section presents the major contribution of the paper, which is to construct a universal process for $\mathbb{VPC}$.
Since a universal process is a special case of an interpreter, we will first give a formal definition of the latter.
We then define an interpreter of $\mathbb{VPC}^{!}$ in $\mathbb{VPC}$.
Finally we modify the definition of the interpreter to produce the desired universal process of $\mathbb{VPC}$.
We hope that this two step construction offers a clearer presentation of our methodology.

Suppose $\mathbb{L},\mathbb{M}$ are complete models.
We intend to formalize the relationship saying that $\mathbb{M}$ is capable of interpreting all the $\mathbb{L}$-processes {\em within} $\mathbb{M}$.
Informally an {\em interpreter} of $\mathbb{L}$ in $\mathbb{M}$ is an $\mathbb{M}$-process such that after inputting a G\"{o}del number of an $\mathbb{L}$-process it behaves like the $\mathbb{L}$-process represented by the number.
A prerequisite for the existence of such an interpreter is that $\mathbb{M}$ should be at least as expressive as $\mathbb{L}$.
This is because if $\mathbb{L}\not\sqsubseteq\mathbb{M}$ then there is an $\mathbb{L}$-process whose interactive behavior cannot be simulated by any $\mathbb{M}$-processes.
When this is the case there cannot be any interpreter of $\mathbb{L}$ in $\mathbb{M}$.
We conclude that every interpreter of $\mathbb{L}$ in $\mathbb{M}$ is based on an expressiveness relation $\propto$ from $\mathbb{L}$ to $\mathbb{M}$.

What is expected of an interpreter?
There is no point for it to simulate a term containing free variables.
But it is expected to be able to manipulate bound variables since they only act as placeholders.
An interpreter can deal with a finite number of global names.
But no interpreter can store an infinite number of global names.
The issue concerning local names is more tricky.
Different models have different naming policies.
Some models admit dynamic creation of local names, others do not.
So in general an interpreter must know the number of distinct local names appearing in a process in order to simulate it properly.
Talking about the number of distinct names, we would like to emphasize that $(b)(\overline{a}b.\overline{b}\,|\,(b)(c)\overline{a}b.\overline{a}c.\overline{b}.\overline{c})$ contains two, not three, distinct local names, although semantically there are three local names.
This static view is important for G\"{o}delization.
Let $\mathcal{N}^{*}$ be the set of finite lists of names, ranged over by $\jmath$.
The notation $a\in\jmath$ means that $a$ appears in $\jmath$.
We have the following description of an interpreter:
\begin{quote}
An interpreter of $\mathbb{L}$ in $\mathbb{M}$ is a family $\{\mathcal{I}_{c}^{i,\jmath}\}_{i\in\textsf{N},\jmath\in\mathcal{N}^{*},c\notin\jmath}$ of $\mathbb{M}$-processes such that, for all $i\in\textsf{N}$, $\jmath\in\mathcal{N}^{*}$ and $c\in\mathcal{N}$ such that $c\notin\jmath$, after picking up a G\"{o}del number at channel $c$ the process $\mathcal{I}_{c}^{i,\jmath}$ can simulate all $\mathbb{L}$-processes that have at most $i$ distinct local names and contain no more global names than those appearing in $\jmath$.
\end{quote}
We will write $\mathcal{I}_{c}^{i,a_{1}\ldots a_{k}}$ if $\jmath$ is the list $a_{1}\ldots a_{k}$.
The superscript $i$ is often omitted.
The interpretation of a number by $\mathcal{I}_{c}^{i,\jmath}$ differs from the interpretation of the same number by $\mathcal{I}_{c}^{i,\jmath'}$ in that they have different interfaces. However $\mathcal{I}_{c}^{i,\jmath}$ and $\mathcal{I}_{c}^{i',\jmath}$ may produce the same interpretation of a number if the number encodes a process that has at most $\min\{i,i'\}$ local names.

Now suppose $k$ is the G\"{o}del index of an $\mathbb{L}$-process $P$ whose set of global names is a subset of $\{a_{1},\ldots,a_{j}\}$ and whose number of local names is no more than $i$.
Let $\{\llbracket\_\rrbracket_{c}\}_{c\in\mathcal{N}}$ be an indexed encoding function of the natural numbers in $\mathbb{M}$.
The process $\mathcal{I}_{c}^{i,a_{1}\ldots a_{j}}$ must satisfy the following property: There exists a unique $Q$ such that
\begin{equation}\label{2011-06-14}
\llbracket k\rrbracket_{c}\,|\,\mathcal{I}_{c}^{i,a_{1}\ldots a_{j}}\stackrel{\iota}{\longrightarrow}\;Q\propto^{-1}P,
\end{equation}
where $\propto$ is the subbisimilarity the interpretation is based upon and $\propto^{-1}$ is the inverse relation of $\propto$.
After a single step interaction with $\llbracket k\rrbracket_{c}$ the process $\mathcal{I}_{c}^{i,a_{1}\ldots a_{j}}$ becomes an $\mathbb{M}$-version of $P$ under $\propto$.
Since there may be many subbisimilarities from $\mathbb{L}$ to $\mathbb{M}$ and possibly infinite number of encodings of the natural numbers into $\mathbb{M}$, it is more precise to define an interpreter of $\mathbb{L}$ in $\mathbb{M}$ as the tuple $\langle\{\mathcal{I}_{c}^{i,\jmath}\}_{i\in\textsf{N},\jmath\in\mathcal{N}^{*},c\notin\jmath},\llbracket \_\rrbracket,\propto\rangle$ that satisfies (\ref{2011-06-14}).
This completes the formal definition of interpreter.

Let's write $\mathbb{L}\in\mathbb{M}$ if there is an interpreter of $\mathbb{L}$ in $\mathbb{M}$.
We may think of $\mathbb{L}\in\mathbb{M}$ as an internal version, or a programming version, of $\mathbb{L}\sqsubseteq\mathbb{M}$.
In the terminology of programming language, $\mathbb{L}\in\mathbb{M}$ says that $\mathbb{L}$ can be implemented in $\mathbb{M}$.

The distinction between a translation and an implementation should now be clear.
A translation is a reduction from a source model to a target model.
It is a meta theoretical operation.
An implementation is a family of processes in the target model that is capable of reproducing a process of the source model at will.

\subsection{G\"{o}del Index}\label{sec-Godel-Index}

\begin{figure}[t]
\begin{center}
\begin{tabular}{rcl}
$\llbracket{\bf 0}\rrbracket_{\mathfrak{i}}$ &$\;\stackrel{\rm def}{=}\;$& $0$, \\
$\llbracket a(x).T\rrbracket_{\mathfrak{i}}$ &$\;\stackrel{\rm def}{=}\;$& $7*\langle \varsigma(a),\varsigma(x),\llbracket T\rrbracket_{\mathfrak{i}}\rangle+1$, \\
$\llbracket \overline{a}(t).T\rrbracket_{\mathfrak{i}}$ &$\;\stackrel{\rm def}{=}\;$& $7*\langle \varsigma(a),\llbracket t\rrbracket_{\varsigma},\llbracket T\rrbracket_{\mathfrak{i}}\rangle+2$, \\
$\llbracket T\,|\,T'\rrbracket_{\mathfrak{i}}$ &$\;\stackrel{\rm def}{=}\;$& $7*\langle \llbracket T\rrbracket_{\mathfrak{i}},\llbracket T'\rrbracket_{\mathfrak{i}}\rangle+3$, \\
$\llbracket (c)T\rrbracket_{\mathfrak{i}}$ &$\;\stackrel{\rm def}{=}\;$& $7*\langle \varsigma(c),\llbracket T\rrbracket_{\mathfrak{i}}\rangle+4$, \\
$\llbracket \textit{if}\;\varphi\;\textit{then}\;T\rrbracket_{\mathfrak{i}}$ &$\;\stackrel{\rm def}{=}\;$& $7*\langle \llbracket\varphi\rrbracket_{\varsigma},\llbracket T\rrbracket_{\mathfrak{i}}\rangle+5$, \\
$\llbracket !a(x).T\rrbracket_{\mathfrak{i}}$ &$\;\stackrel{\rm def}{=}\;$& $7*\langle \varsigma(a),\varsigma(x),\llbracket T\rrbracket_{\mathfrak{i}}\rangle+6$, \\
$\llbracket !\overline{a}(t).T\rrbracket_{\mathfrak{i}}$ &$\;\stackrel{\rm def}{=}\;$& $7*\langle \varsigma(a),\llbracket t\rrbracket_{\varsigma},\llbracket T\rrbracket_{\mathfrak{i}}\rangle+7$. \\
\end{tabular}
  \caption{G\"{o}del Index of $\mathbb{VPC}^{!}$-Term.}
  \label{G-Index}
\end{center}
\end{figure}

We avail ourselves of an effective bijective function \[\langle\_,\ldots,\_\rangle_{k}:\underset{k\ \mathrm{times}}{\underbrace{\textsf{N}\times\ldots\times\textsf{N}}}\rightarrow\textsf{N},\] whose inverse function is composed of the unary functions $(\_)_{0},\ldots,(\_)_{k-1}$.
For clarity we sometimes write for instance $z_{i,j}$ for $((z)_{i})_{j}$ when no confusion arises.
We assume that $\langle0,0,\ldots,0\rangle_{k}=0$, and we often omit the subscript in $\langle\_,\ldots,\_\rangle_{k}$.
For convenience we assume that the unary pairing function is the identity function and the 0-ary pairing function is the constant $0$.

By abusing notations, let $\varsigma$ denote both a bijective function from the set $\mathcal{N}$ of names to $\textsf{N}$ and a bijective function from the set $\textsf{V}$ of variables to $\textsf{N}$.
Using $\varsigma$ as an oracle function, we can define an effective bijective function from the set $\textsf{T}$ of terms to $\textsf{N}$ and an effective bijective function from the set $\textsf{B}$ of formulas to $\textsf{N}$.
We will denote both by $\llbracket\_\rrbracket_{\varsigma}$ and omit the obvious structural definition.

Using the standard technique the G\"{o}del number of a $\mathbb{VPC}^{!}$-term is defined by the function $\llbracket\_\rrbracket_{\mathfrak{i}}$ introduced in Figure~\ref{G-Index}.
The function $\llbracket\_\rrbracket_{\mathfrak{i}}$ is a {\em bijection} between the set of the $\mathbb{VPC}^{!}$-terms and the set of the natural numbers.
It should be emphasized that we prohibit the use of $\alpha$-conversion when we are assigning G\"{o}del numbers to the $\mathbb{VPC}^{!}$-terms.
The encodings of say $a(x).a(x).\overline{b}(x)$ and $a(x).a(y).\overline{b}(y)$ are different, even though they are treated as syntactically the same term when $\alpha$-conversion is admitted.

We get another set of G\"{o}del indices if we use an oracle function different from $\varsigma$.
The definition of our interpreter does not depend on any particular choice of such a function.
For a $\mathbb{VPC}^{!}$-process $P$ using $k$ global names $a_{1},\ldots,a_{k}$ and $i$ local names, a {\em normal index} of $P$ is the one in which the global names $a_{1},\ldots,a_{k}$ are indexed by $1,2,\ldots,k$ and the local names are indexed by $k+1,k+2,\ldots,k+i$.

\subsection{A VPC Interpreter}

Having fixed the G\"{o}del index for $\mathbb{VPC}^{!}$, we are ready to define an interpreter of $\mathbb{VPC}^{!}$ in $\mathbb{VPC}$.
The new problem with such an interpretation is how to simulate interactions properly.
The proofs of Turing equivalence in computation theory literature boil down to establishing effective codivergent bisimulation relations from source models to target models.
We extend this approach to encodings that take care of equipollence and extensionality.
A major difference between an encoding of one computation model into another and an encoding for interaction models is that the former is a translator whereas the latter is an interpreter.
Since a process normally receives messages from environment in its life cycle, and the behaviour of the process after an input action depends on the information content of received message, the translation of a component of the process cannot be fixed in a static fashion.
It has to be done during execution, hence the on-the-fly style of interpretation.

The interpreter is a family of processes $\{\mathcal{I}^{i,\jmath}_{d}\}_{i\in\textsf{N},\jmath\in\mathcal{N}^{*},d\notin\jmath}$.
The definition of $\mathcal{I}_{d}^{i,a_{1}\ldots a_{k}}$ is given by
\begin{eqnarray}\label{2011-06-27}
\mathcal{I}^{i,a_{1}\ldots a_{k}}_{d} &=& d(z).(h)(\mathcal{P}_{\mathfrak{i}}(z)\,|\,h(z).\mathcal{S}_{\mathfrak{i}}(z)).
\end{eqnarray}
The interpreter executes the two subroutines sequentially.
\begin{itemize}
\item If $z$ is the G\"{o}del number of a process, say $A$, whose global names are among $a_{1}\ldots a_{k}$ and the number of the local names appearing in it is bounded by $i$, the {\em parser} $\mathcal{P}_{\mathfrak{i}}(z)$ transforms $z$ to a normal G\"{o}del index $z'$ that codes up a process $\alpha$-convertible to $A$, and then releases $z'$ through the channel $h$.
If $z$ is illegitimate the parser aborts the interpretation.
In other words $\mathcal{I}_{d}^{i,a_{1}\ldots a_{k}}$ chooses to interpret the index of a process of a wrong type as an index for ${\bf 0}$.

\item The {\em simulator} $\mathcal{S}_{\mathfrak{i}}(z')$ operates on the received normal G\"{o}del number and simulates the process indexed by the number.
\end{itemize}
The interpretation makes use of the following recursive functions $\textsf{r}_{7},\textsf{d}_{7}$:
\begin{eqnarray*}\label{2011-06-17}
\textsf{r}_{7}(z) &\stackrel{\rm def}{=}& \left\{\begin{array}{ll}
0, & \ \mathrm{if}\ z=0, \\
i, & \ \mathrm{if}\ 1\le i\le7\ \mathrm{and}\ \exists j.z=7*j+i, \\
\end{array}\right. \\
\textsf{d}_{7}(z) &\stackrel{\rm def}{=}& \left\{\begin{array}{ll}
0, & \ \mathrm{if}\ z=0, \\
j, & \ \mathrm{if}\ \exists i{\in}\{1,\ldots,7\}.z=7*j+i. \\
\end{array}\right.
\end{eqnarray*}
The operation carried out by the parser is purely arithmetical.
So it can be implemented in $\mathbb{VPC}$.
Let's however take a look at an outline of the following implementation of $\mathcal{P}_{\mathfrak{i}}(z)$:
\begin{equation}\label{2016-04-04}
(g_{1}\ldots g_{k+i})(c)(e)(\prod_{j=1}^{k+i}\overline{g_{j}}(0)\,|\,\overline{c}(1)\,|\,\mathcal{G}_{\mathfrak{i}}(z,0)\,|\,e.\mathcal{N}_{\mathfrak{i}}(z,0)),
\end{equation}
where $\prod_{j=1}^{k+i}\overline{g_{j}}(0)$ stands for $\overline{g_{1}}(0)\,|\,\ldots\,|\,\overline{g_{k+i}}(0)$.
The {\em grammar checker} $\mathcal{G}_{\mathfrak{i}}(z,v)$ is defined in Fig.~\ref{G4VPC}.
\begin{figure}[t]
\begin{center}
\begin{tabular}{l}
 ${\bf case}\ z\ {\bf of}$ \\
 $\ \ \textsf{r}_{7}(z){=}0 \;\Rightarrow\; c(x).(\textit{if}\;x=1\;\textit{then}\;\overline{e}\;\textit{else}\;\overline{c}(x-1))$; \\
 $\ \ \textsf{r}_{7}(z){=}1 \;\Rightarrow\; \textit{if}\;\ L_1(\textsf{d}_{7}(z)_{0})\;\textit{then}\;\mathcal{G}_{\mathfrak{i}}(\textsf{d}_{7}(z)_{2},\langle\textsf{d}_{7}(z)_{1}+1,v\rangle)$; \\
 $\ \ \textsf{r}_{7}(z){=}2 \;\Rightarrow\; \textit{if}\;L_2^1(\textsf{d}_{7}(z)_{1})\;\textit{then}\; \textit{if}\;L_2^0(\textsf{d}_{7}(z)_{0})\;\textit{then}\;\mathcal{G}_{\mathfrak{i}}(\textsf{d}_{7}(z)_{2},v)$; \\
 $\ \ \textsf{r}_{7}(z){=}3 \;\Rightarrow\; c(x).(\overline{c}(x+1)\,|\,\mathcal{G}_{\mathfrak{i}}(\textsf{d}_{7}(z)_{0},v)\,|\,\mathcal{G}_{\mathfrak{i}}(\textsf{d}_{7}(z)_{1},v))$; \\
 $\ \ \textsf{r}_{7}(z){=}4 \;\Rightarrow\; \textit{if}\;\ L_4(\textsf{d}_{7}(z)_{0})\;\textit{then}\;\mathcal{G}_{\mathfrak{i}}(\textsf{d}_{7}(z)_{1},v)$; \\
 $\ \ \textsf{r}_{7}(z){=}5 \;\Rightarrow\; \textit{if}\;\ L_5(\textsf{d}_{7}(z)_{0})\;\textit{then}\;\mathcal{G}_{\mathfrak{i}}(\textsf{d}_{7}(z)_{1},v)$; \\
 $\ \ \textsf{r}_{7}(z){=}6 \;\Rightarrow\; \textit{if}\;\ L_6(\textsf{d}_{7}(z)_{0})\;\textit{then}\;\mathcal{G}_{\mathfrak{i}}(\textsf{d}_{7}(z)_{2},\langle\textsf{d}_{7}(z)_{1}+1,v\rangle)$; \\
 $\ \ \textsf{r}_{7}(z){=}7 \;\Rightarrow\; \textit{if}\;L_7^1(\textsf{d}_{7}(z)_{1})\;\textit{then}\; \textit{if}\;L_7^0(\textsf{d}_{7}(z)_{0})\;\textit{then}\;\mathcal{G}_{\mathfrak{i}}(\textsf{d}_{7}(z)_{2},v)$ \\
 ${\bf end}\ {\bf case}$. \\
\end{tabular}
  \caption{Grammar Checker $\mathcal{G}_{\mathfrak{i}}(z,v)$.}
  \label{G4VPC}
\end{center}
\end{figure}
It aborts the interpreter if any one of the following happens:
\begin{itemize}
\item The number of the indices of global names is more than $k$.
\item The number of the indices of local names is more than $i$.
\item There is an index for a free variable.
\end{itemize}
At the name $c$ is recorded the number of the concurrent components the parser has encountered.
Initially there is only one component, which explains the presence of $\overline{c}(1)$.
The parser ends successfully if in the end the value at $c$ is $0$.
The names $g_{1},\ldots,g_{k}$ are used to store the indices for the global names and the names $g_{k+1},\ldots,g_{k+i}$ for the local names.
If $j_{1}$ is the first index for a local name $\mathcal{G}_{\mathfrak{i}}(z,v)$ encounters, the grammar checker stores $j_{1}+1$ at $g_{k+1}$.
This can be done by invoking $g_{k+1}(x).\overline{g_{k+1}}(j_{1}+1)$.
Similarly if $j_{2}$ is the second index of a local name $\mathcal{G}_{\mathfrak{i}}(z,v)$ encounters, then $\mathcal{G}_{\mathfrak{i}}(z,v)$ stores $j_{2}+1$ at $g_{k+2}$.
In completely the same fashion, the grammar checker stores the G\"{o}del numbers that represent the global names at $g_{1},\ldots,g_{k}$ in the order they are discovered.
The second parameter $v$ of $\mathcal{G}_{\mathfrak{i}}(z,v)$ codes up the bound variables already discovered.
If for example the bound variables are $x_{1},\ldots,x_{k'}$ encountered in that order then $v$ would be $\langle \varsigma(x_{k'})+1,\langle \varsigma(x_{k'-1})+1,\ldots,\langle \varsigma(x_{1})+1,0\rangle\ldots\rangle\rangle$.
The Boolean functions $L_1,L_2^1,L_2^0,L_4,L_5,L_6,L_7^1,L_7^0$ check if the number of names is under the limit or if all variables are bound.
We now explain how $\mathcal{G}_{\mathfrak{i}}(z,v)$ works.
\begin{itemize}
\item
$\textsf{r}_{7}(z)=0$.
If the number of concurrent components becomes zero, end $\mathcal{G}_{i}$ successfully and initiate the {\em normalizer} $\mathcal{N}_{\mathfrak{i}}(z,0)$.

\item
$\textsf{r}_{7}(z)=1$.
The number $\textsf{d}_{7}(z)_{1}$ is the index of a bound variable.
The process $L_1(\textsf{d}_{7}(z)_{0})$ needs to make sure that if $\textsf{d}_{7}(z)_{0}$ is neither the index of a local name nor an index of a global name that has been recorded before, then the number stored at $g_{k}$ must be $0$.
If $L_1(\textsf{d}_{7}(z)_{0})$ succeeds, the number $\textsf{d}_{7}(z)_{0}+1$ is stored at the least $g_j$ that has not been used.
After $L_1(\textsf{d}_{7}(z)_{0})$ has ended successfully, the number $\textsf{d}_{7}(z)_{1}+1$ is added to the list of the indices of the bound variables already parsed and is passed down recursively.

\item
$\textsf{r}_{7}(z)=2$.
The subroutine $L_2^1(\textsf{d}_{7}(z)_{1})$ checks if the term represented by the number $\textsf{d}_{7}(z)_{1}$ contains an unknown variable.
It aborts if it encounters an index that does not appear in the tuple encoded by $v$.
If $L_2^1(\textsf{d}_{7}(z)_{1})$ succeeds, $L_2^0(\textsf{d}_{7}(z)_{0})$ checks the legitimacy of the encoding $\textsf{d}_{7}(z)_{0}$.
The process $L_2^0$ works in the same manner as the process $L_1$.

\item
$\textsf{r}_{7}(z)=3$.
The counter at $c$ is incremented by $1$ since one concurrent component is split into two.

\item
$\textsf{r}_{7}(z)=4$.
The subroutine $L_4(\textsf{d}_{7}(z)_{0})$ checks if the number $\textsf{d}_{7}(z)_{0}+1$ is the same as the number stored at some $g_{j}$, where $k+1\le j\le k+i$.
If the answer is positive, $L_4(\textsf{d}_{7}(z)_{0})$ succeeds; otherwise it checks if the number at $g_{k+i}$ is $0$.
If the answer to the latter query is negative, it aborts; otherwise it succeeds after it has stored the number $\textsf{d}_{7}(z)_{0}+1$ at the appropriate $g_{j}$.

\item
$\textsf{r}_{7}(z)=5$.
The subroutine $L_5(\textsf{d}_{7}(z)_{0})$ checks if the number $\textsf{d}_{7}(z)_{0}$ codes up a well formed formula.
Specifically it needs to make sure that the formula coded up by the number does not contain any free variable.

\item
$\textsf{r}_{7}(z)=6$, $\textsf{r}_{7}(z)=7$.
These cases are similar to the cases $\textsf{r}_{7}(z)=1$, $\textsf{r}_{7}(z)=2$ respectively.
\end{itemize}
After $\mathcal{G}_{\mathfrak{i}}(z,0)$ has ended successfully, it starts the process $\mathcal{N}_{\mathfrak{i}}(z,0)$.

\begin{figure}[t]
\begin{center}
\begin{tabular}{l}
 ${\bf case}\ z\ {\bf of}$ \\
 $\ \ \textsf{r}_{7}(z){=}0 \;\Rightarrow\; \overline{h}(0)$; \\
 $\ \ \textsf{r}_{7}(z){=}1 \;\Rightarrow\; Find(\textsf{d}_{7}(z)_{0},w,y).(b)(b(x).\overline{h}(7{*}\langle y,\textsf{d}_{7}(z)_{1},x\rangle{+}1)\,|\,(h)(h(u).\overline{b}(u)\,|\,\mathcal{N}_{\mathfrak{i}}(\textsf{d}_{7}(z)_{2},w)))$; \\
 $\ \ \textsf{r}_{7}(z){=}2 \;\Rightarrow\; Find(\textsf{d}_{7}(z)_{0},w,y).(b)(b(x).\overline{h}(7{*}\langle y,\textsf{d}_{7}(z)_{1},x\rangle{+}2)\,|\,(h)(h(u).\overline{b}(u)\,|\,\mathcal{N}_{\mathfrak{i}}(\textsf{d}_{7}(z)_{2},w)))$; \\
 $\ \ \textsf{r}_{7}(z){=}3 \;\Rightarrow\; (b_{1}b_{2})(b_{1}(x_{1}).b_{2}(x_{2}).\overline{h}(7{*}\langle x_{1},x_{2}\rangle{+}3)$ \\
 $\ \ \ \ \ \ \ \ \ \ \ \ \ \ \ \ \ \ \ \ \ \ \ \ \ \ \ \,|\,(h)(h(u).\overline{b_{1}}(u)\,|\,\mathcal{N}_{\mathfrak{i}}(\textsf{d}_{7}(z)_{0},w))\,|\,(h)(h(u).\overline{b_{2}}(u)\,|\,\mathcal{N}_{\mathfrak{i}}(\textsf{d}_{7}(z)_{1},w)))$; \\
 $\ \ \textsf{r}_{7}(z){=}4 \;\Rightarrow\; (b)(b(x).\overline{h}(7{*}\langle \textsf{d}_{7}(z)_{0},x\rangle{+}4)\,|\,(h)(h(u).\overline{b}(u)\,|\,\mathcal{N}_{\mathfrak{i}}(\textsf{d}_{7}(z)_{1},\langle \textsf{d}_{7}(z)_{0},w\rangle)))$; \\
 $\ \ \textsf{r}_{7}(z){=}5 \;\Rightarrow\; (b)(b(x).\overline{h}(7{*}\langle \textsf{d}_{7}(z)_{0},x\rangle{+}5)\,|\,(h)(h(u).\overline{b}(u)\,|\,\mathcal{N}_{\mathfrak{i}}(\textsf{d}_{7}(z)_{1},w)))$; \\
 $\ \ \textsf{r}_{7}(z){=}6 \;\Rightarrow\; Find(\textsf{d}_{7}(z)_{0},w,y).(b)(b(x).\overline{h}(7{*}\langle y,\textsf{d}_{7}(z)_{1},x\rangle{+}6)\,|\,(h)(h(u).\overline{b}(u)\,|\,\mathcal{N}_{\mathfrak{i}}(\textsf{d}_{7}(z)_{2},w)))$; \\
 $\ \ \textsf{r}_{7}(z){=}7 \;\Rightarrow\; Find(\textsf{d}_{7}(z)_{0},w,y).(b)(b(x).\overline{h}(7{*}\langle y,\textsf{d}_{7}(z)_{1},x\rangle{+}7)\,|\,(h)(h(u).\overline{b}(u)\,|\,\mathcal{N}_{\mathfrak{i}}(\textsf{d}_{7}(z)_{2},w)))$ \\
 ${\bf end}\ {\bf case}$ \\
\end{tabular}
  \caption{Normalizer $\mathcal{N}_{\mathfrak{i}}(z,w)$.}
  \label{N4VPC}
\end{center}
\end{figure}

Using the values stored at $g_{j}$'s,
$\mathcal{N}_{\mathfrak{i}}(z,0)$ transforms G\"{o}del index $z$ to a normal G\"{o}del index $z'$ and then releases $z'$ at $h$.
An implementation of $\mathcal{N}_{\mathfrak{i}}$ is given in Fig.~\ref{N4VPC}.
The operation $Find(j,w,y)$ returns as the value of $y$ the $j'$ such that $j$ is the number stored at $g_{j'}$.
If $j$ appears in $w$, then look for the index $j'$ in $\{k+1,\ldots,k+i\}$; otherwise look for the $j'$ in $\{1,\ldots,k\}$.
In $w$ are stored the local names that might appear in the term being processed.
This additional complexity is due to the fact that a number could denote both a local name and a global name.
By making use of dynamic binding, a recursive invocation of $\mathcal{N}_{\mathfrak{i}}(z,w)$ collects through the localized channel $h$ the normal encoding of a subterm and passes it to its caller using the local channel $b$.

\begin{figure}[t]
\begin{center}
\begin{tabular}{l}
 ${\bf case}\ z\ {\bf of}$ \\
 $\ \ \textsf{r}_{7}(z){=}0 \;\Rightarrow\; {\bf 0}$; \\
 $\ \ \textsf{r}_{7}(z){=}1 \;\Rightarrow\; Nth(\textsf{d}_{7}(z)_{0},y).a_{y}(x).\mathcal{S}_{\mathfrak{i}}([x/\textsf{d}_{7}(z)_{1}]\textsf{d}_{7}(z)_{2})$; \\
 $\ \ \textsf{r}_{7}(z){=}2 \;\Rightarrow\; Nth(\textsf{d}_{7}(z)_{0},y).\overline{a_{y}}(val(\textsf{d}_{7}(z)_{1})).\mathcal{S}_{\mathfrak{i}}(\textsf{d}_{7}(z)_{2})$; \\
 $\ \ \textsf{r}_{7}(z){=}3 \;\Rightarrow\; \mathcal{S}_{\mathfrak{i}}(\textsf{d}_{7}(z)_{0})\,|\,\mathcal{S}_{\mathfrak{i}}(\textsf{d}_{7}(z)_{1})$; \\
 $\ \ \textsf{r}_{7}(z){=}4 \;\Rightarrow\; Nth(\textsf{d}_{7}(z)_{0},y).(a_{y})\mathcal{S}_{\mathfrak{i}}(\textsf{d}_{7}(z)_{1})$; \\
 $\ \ \textsf{r}_{7}(z){=}5 \;\Rightarrow\; \textit{if}\;val(\textsf{d}_{7}(z)_{0})\;\textit{then}\;\mathcal{S}_{\mathfrak{i}}(\textsf{d}_{7}(z)_{1})$; \\
 $\ \ \textsf{r}_{7}(z){=}6 \;\Rightarrow\; Nth(\textsf{d}_{7}(z)_{0},y).!a_{y}(x).\mathcal{S}_{\mathfrak{i}}([x/\textsf{d}_{7}(z)_{1}]\textsf{d}_{7}(z)_{2})$; \\
 $\ \ \textsf{r}_{7}(z){=}7 \;\Rightarrow\; Nth(\textsf{d}_{7}(z)_{0},y).!\overline{a_{y}}(val(\textsf{d}_{7}(z)_{1})).\mathcal{S}_{\mathfrak{i}}(\textsf{d}_{7}(z)_{2})$ \\
 ${\bf end}\ {\bf case}$ \\
\end{tabular}
  \caption{Simulator $\mathcal{S}_{\mathfrak{i}}(z)$.}
  \label{Ivpc}
\end{center}
\end{figure}

The simulator $\mathcal{S}_{\mathfrak{i}}(z)$, defined in Fig.~\ref{Ivpc}, simulates the $\mathbb{VPC}^{!}$-process coded up by $z$ in an on-the-fly fashion.
The arithmetical operations referred to in Fig.~\ref{Ivpc} are described below:
\begin{itemize}
\item
The notation $[x/\textsf{d}_{7}(z)_{1}]\textsf{d}_{7}(z)_{2}$ stands for the G\"{o}del number obtained from $\textsf{d}_{7}(z)_{2}$ by substituting $x$, which must have been instantiated by an input action at the moment this operation is executed, for $\textsf{d}_{7}(z)_{1}$.
\item
The notation $val(\textsf{d}_{7}(z)_{1})$ denotes the result of evaluating the term expression coded up by $\textsf{d}_{7}(z)_{1}$.
Similarly $val(\textsf{d}_{7}(z)_{0})$ denotes the result of evaluating the formula coded up by $\textsf{d}_{7}(z)_{0}$.
Notice that when the evaluation operations start, neither $\textsf{d}_{7}(z)_{0}$ nor $\textsf{d}_{7}(z)_{1}$ contains any variables.
\end{itemize}
The prefix operation $Nth(\textsf{d}_{7}(z)_{0},y)$ returns $j$ as the value of $y$ if $\textsf{d}_{7}(z)_{0}$ is stored at $g_{j}$, where $1\le j\le k+i$.
Some comments on $\mathcal{S}_{\mathfrak{i}}(z)$ are as follows:
\begin{itemize}
\item
$\textsf{r}_{7}(z){=}1$.
The continuation $a_{y}(x).\mathcal{S}_{\mathfrak{i}}([x/\textsf{d}_{7}(z)_{1}]\textsf{d}_{7}(z)_{2})$ is an abbreviation of
\[\textit{if}\;y{=}1\;\textit{then}\;a_{1}(x).\mathcal{S}_{\mathfrak{i}}([x/\textsf{d}_{7}(z)_{1}]\textsf{d}_{7}(z)_{2})\,|\,\ldots\,|\, \textit{if}\;y{=}k\;\textit{then}\;a_{k}(x).\mathcal{S}_{\mathfrak{i}}([x/\textsf{d}_{7}(z)_{1}]\textsf{d}_{7}(z)_{2}).\]
We need to know the number of local names in advance in order to define the above process.
If we hope to drop the superscript $i$ from $\mathcal{I}^{i,a_{1}\ldots a_{k}}_{d}$ defined in (\ref{2011-06-27}) we need to look for a very different interpretation.
In case $\textsf{r}_{7}(z){=}2$, the subterm $\overline{a_{y}}(val(\textsf{d}_{7}(z)_{1})).\mathcal{S}_{\mathfrak{i}}(\textsf{d}_{7}(z)_{2})$ is defined similarly.
\item
$\textsf{r}_{7}(z){=}4$.
This case deserves special attention.
The simulator makes essential use of the dynamic binding plus recursive definition.
The number in $\mathcal{S}_{\mathfrak{i}}(\textsf{d}_{7}(z)_{1})$ that encodes $a_y$ is decoded into $a_y$, and the continuation is restricted by the operator $(a_y)$ dynamically.
An implementation of the recursive call of the simulator in $\mathbb{VPC}^{!}$ would render the localization operator $(a_{y})$ detached from the body to which the operator should apply.
For this reason the present encoding does not work in $\mathbb{VPC}^{!}$.
\end{itemize}
It is remarkable that the interpreter uses only one dummy variable $x$.
No confusion among the bound variables can ever arise.

Although we have not supplied all the details of the interpretation, the key ingredients that support the following claim have all been spelled out.

\begin{thm}\label{VPCeinVPC}
$\mathbb{VPC}^{!} \in \mathbb{VPC}$.
\end{thm}
\begin{proof}
The argument given here rests on the completeness of $\mathbb{VPC}$~\cite{FuYuxi-2013-VPC} and our trust in the Church-Turing Thesis.
We summarize the main points below:
\begin{itemize}
\item
The encoding of the natural numbers is given by the class $\{\overline{a}(k)\}_{k\in\textsf{N},a\in\mathcal{N}}$.
This encoding is correct with respect to the absolute equality $=_{\mathbb{VPC}}$~\cite{FuYuxi-2013-VPC}.
\item
All the number theoretical operations involved in the definition of $\mathcal{I}^{i,\jmath}_{d}$, specifically the parser defined in~(\ref{2016-04-04}), are computable.
It follows that these operations are all definable in $\mathbb{VPC}$.

\item
The relation $\propto_{\mathfrak{i}}:\mathbb{VPC}^{!} \rightarrow \mathbb{VPC}$ is basically the structural embedding composed with the equality relation $=_{\mathbb{VPC}}$, where the replication operator in $\mathbb{VPC}^{!}$ is interpreted by the derived replication operator defined in~(\ref{2016-04-03-a}) and~(\ref{2016-04-03-b}).

\item
Let $\mathfrak{I}$ denote the composition $\mathfrak{S};=_{\mathbb{VPC}}$, where $\mathfrak{S}$ denotes the following relation
\[\left\{(P,(d)(\overline{d}(k)\,|\,\mathcal{I}^{i,\jmath}_{d})) \left|
\begin{array}{l}
P\ \mathrm{is}\ \mathrm{a}\ \mathbb{VPC}^{!}\ \mathrm{process},\
\mathrm{all}\ \mathrm{global}\ \mathrm{names}\ \mathrm{of}\ P \\
\mathrm{are}\ \mathrm{in}\ \jmath,\ \mathrm{the}\ \mathrm{number}\ \mathrm{of}\ \mathrm{local}\ \mathrm{names}\ \mathrm{of}\ P\ \mathrm{is} \\
\mathrm{no}\ \mathrm{more}\ \mathrm{than}\ i,\ k\ \mathrm{is}\ \mathrm{an}\ \mathrm{index}\ \mathrm{of}\ P,\ \mathrm{and}\ d\notin\jmath.
\end{array}
\right\}.\right.\]
To establish (\ref{2011-06-14}) it is sufficient to demonstrate that $\mathfrak{I}$ is a subbisimilarity.
It is easy to argue informally that the definition of the simulator in Fig.~\ref{Ivpc} renders true the following statements:
\begin{itemize}
\item Since all the actions of $P$ are also actions of $(d)(\overline{d}(k)\,|\,\mathcal{I}^{i,\jmath}_{d})$, the following explicit bisimulation property is valid whenever $P\mathfrak{S}M=_{\mathbb{VPC}}Q$:
\begin{itemize}
\item If $P\stackrel{\alpha}{\longrightarrow}P'$ then $Q\rightarrow^{*}\stackrel{\alpha}{\longrightarrow}Q'=_{\mathbb{VPC}}M'\mathfrak{S}^{-1}P'$ for some $Q'$ and $M'$.
\item If $Q\stackrel{\alpha}{\longrightarrow}Q'$ and $\alpha\ne\tau$, then $P\stackrel{\alpha}{\longrightarrow}P'\mathfrak{S}M'=_{\mathbb{VPC}}Q'$ for some $P',M'$.
\item If $Q\stackrel{\iota}{\longrightarrow}Q'$ then $P\stackrel{\iota}{\longrightarrow}P'\mathfrak{S}M'=_{\mathbb{VPC}}Q'$ for some $P',M'$.
\item If $Q\rightarrow Q'$ then either $P\mathfrak{S}M=_{\mathbb{VPC}}Q'$ or $P\rightarrow P'\mathfrak{S}M'=_{\mathbb{VPC}}Q'$ for some $P',M'$.
\end{itemize}
\item $\mathfrak{I}$ is codivergent since the extra number theoretical manipulations do not introduce any divergence.
\end{itemize}
These properties are enough for us to conclude that $\mathfrak{I}$ is a subbisimilarity.
\end{itemize}
This completes the proof sketch.
\end{proof}

\subsection{Universal VPC Process}\label{sec-Unversal-VPC-Process}

A self-interpreter for $\mathbb{M}$ is an interpreter of $\mathbb{M}$ in $\mathbb{M}$.
Such an interpreter is based on a subbisimilarity from $\mathbb{M}$ to $\mathbb{M}$.
In general there is more than one subbisimilarity from $\mathbb{M}$ to $\mathbb{M}$~\cite{FuYuxi},
among which the absolute equality $=_{\mathbb{M}}$ offers a canonical relation in the sense that every process is interpreted by itself.
An interpreter of $\mathbb{M}$ in $\mathbb{M}$ based on the absolute equality $=_{\mathbb{M}}$ is called a {\em universal process} for $\mathbb{M}$.
We will write $\langle\{\mathcal{U}_{c}^{i,\jmath}\}_{i\in\textsf{N},\jmath\in\mathcal{N}^{*},c\notin\jmath},\llbracket \_\rrbracket\rangle$ for a universal process of $\mathbb{M}$.
For all $i,\jmath,c$ the process $\mathcal{U}_{c}^{i,\jmath}$ must satisfy the following property: If $P$ is of the right type and $k$ is a G\"{o}del index of $P$ then a unique $Q$ exists such that
\begin{equation}\label{2011-07-06}
\llbracket k\rrbracket_{c}\,|\,\mathcal{U}_{c}^{i,\jmath}\stackrel{\iota}{\longrightarrow}\;Q=_{\mathbb{M}}P.
\end{equation}
The aim of this subsection is to modify the interpreter constructed in the previous subsection to obtain a universal process for $\mathbb{VPC}$.

We need to explain how parametric definitions are treated.
Now every $\mathbb{VPC}$-term refers to only a finite number of parametric definitions.
Suppose the parametric definitions appearing in $T$ are given by the following equations:
\begin{equation}\label{2011-09-03}
\begin{array}{rcl}
D_{1}(x_{11},\ldots,x_{1i_{1}}) &=& T_{1}, \\
 &\vdots& \\
D_{k}(x_{k1},\ldots,x_{ki_{k}}) &=& T_{k},
\end{array}
\end{equation}
for some $k\ge0$. When $k=0$ we understand that $T$ contains no occurrence of any parametric definition.
The G\"{o}del index $\llbracket T\rrbracket_{\mathfrak{u}}$ of $T$ is defined as follows:
\begin{eqnarray}\label{2011-07-08}
\llbracket T\rrbracket_{\mathfrak{u}} \stackrel{\rm def}{=} \langle k,\langle \llbracket T\rrbracket_{\mathfrak{d}},\llbracket T\rrbracket_{\mathfrak{p}}\rangle\rangle.
\end{eqnarray}
The components of (\ref{2011-07-08}) have the following readings:
\begin{itemize}
\item $k$ is the number of parametric definitions in (\ref{2011-09-03}).
According to our notational convention $\llbracket T\rrbracket_{\mathfrak{d}}$, which is a $k$-ary tuple, is $0$ when $k=0$.

\item The index $\llbracket T\rrbracket_{\mathfrak{d}}$ codes up all the parametric definitions in (\ref{2011-09-03}).
It is given by
\begin{eqnarray}\label{2011-11-05}
\llbracket T\rrbracket_{\mathfrak{d}} \stackrel{\rm def}{=}
\begin{array}{c}
\langle\,\langle i_{1},\langle \varsigma(x_{11}),\ldots,\varsigma(x_{1i_{1}}),\llbracket T_{1}\rrbracket_{\mathfrak{p}}\rangle\rangle, \\
\vdots \\
\ \ \ \ \langle i_{k},\langle \varsigma(x_{k1}),\ldots,\varsigma(x_{ki_{k}}),\llbracket T_{k}\rrbracket_{\mathfrak{p}}\rangle\rangle\,\rangle.
\end{array}
\end{eqnarray}
Notice that this is not a self-referential definition.

\item
The structural definition of $\llbracket\_\rrbracket_{\mathfrak{p}}$ is given in Fig.~\ref{G-Index-VPCdef}.
The only thing worth mentioning is that the index of $D_{j}(t_{j1},\ldots,t_{ji_{j}})$ contains the information about the equation in which $D_{j}$ is defined, the number of parameters of $D_{j}$, and all the terms that instantiate the parameters.
\end{itemize}
At top level the indices of all $\mathbb{VPC}$-terms are of the form (\ref{2011-07-08}).
One may think of $\llbracket T\rrbracket_{\mathfrak{p}}$ as the main program and $\llbracket T\rrbracket_{\mathfrak{d}}$ as the subroutines necessary when executing the program.

\begin{figure}[t]
\begin{center}
\begin{tabular}{rcl}
$\llbracket{\bf 0}\rrbracket_{\mathfrak{p}}$ &$\;\stackrel{\rm def}{=}\;$& $0$, \\
$\llbracket a(x).T\rrbracket_{\mathfrak{p}}$ &$\;\stackrel{\rm def}{=}\;$& $6*\langle \varsigma(a),\varsigma(x),\llbracket T\rrbracket_{\mathfrak{p}}\rangle+1$, \\
$\llbracket \overline{a}(t).T\rrbracket_{\mathfrak{p}}$ &$\;\stackrel{\rm def}{=}\;$& $6*\langle \varsigma(a),\llbracket t\rrbracket_{\varsigma},\llbracket T\rrbracket_{\mathfrak{p}}\rangle+2$, \\
$\llbracket T\,|\,T'\rrbracket_{\mathfrak{p}}$ &$\;\stackrel{\rm def}{=}\;$& $6*\langle \llbracket T\rrbracket_{\mathfrak{p}},\llbracket T'\rrbracket_{\mathfrak{p}}\rangle+3$, \\
$\llbracket (c)T\rrbracket_{\mathfrak{p}}$ &$\;\stackrel{\rm def}{=}\;$& $6*\langle \varsigma(c),\llbracket T\rrbracket_{\mathfrak{p}}\rangle+4$, \\
$\llbracket \textit{if}\;\varphi\;\textit{then}\;T\rrbracket_{\mathfrak{p}}$ &$\;\stackrel{\rm def}{=}\;$& $6*\langle \llbracket\varphi\rrbracket_{\varsigma},\llbracket T\rrbracket_{\mathfrak{p}}\rangle+5$, \\
$\llbracket D_{j}(t_{j1},\ldots,t_{jn_{j}})\rrbracket_{\mathfrak{p}}$ &$\;\stackrel{\rm def}{=}\;$& $6*\langle j,\langle n_{j},\langle\llbracket t_{j1}\rrbracket_{\varsigma},\ldots,\llbracket t_{jn_{j}}\rrbracket_{\varsigma}\rangle\rangle\rangle+6$. \\
\end{tabular}
  \caption{G\"{o}del Index of $\mathbb{VPC}$-Term.}
  \label{G-Index-VPCdef}
\end{center}
\end{figure}

Our universal process $\{\mathcal{U}_{d}^{i,\jmath}\}_{i\in\textsf{N},\jmath\in\mathcal{N}^{*},d\notin\jmath}$ for $\mathbb{VPC}$ is defined as follows:
\begin{eqnarray}\label{2011-07-09}
\mathcal{U}_{d}^{i,a_{1}\ldots a_{k}} &\stackrel{\rm def}{=}& d(z).(h)(\mathcal{P}_{\mathfrak{u}}(z)\,|\,h(z).(e)(\mathcal{D}_{\mathfrak{u}}((z)_{0},z_{1,0})\,|\,\mathcal{S}_{\mathfrak{u}}(z_{1,1})))
\end{eqnarray}
for all $i,a_{1},\ldots,a_{k},d$ such that $d\notin\{a_{1},\ldots,a_{k}\}$.
The process $\mathcal{U}_{d}^{i,a_{1}\ldots a_{k}}$ is similar to $\mathcal{I}_{d}^{i,a_{1}\ldots a_{k}}$ defined in (\ref{2011-06-27}).
We leave out the definition of the parser $\mathcal{P}_{\mathfrak{u}}(z)$ since it is similar to $\mathcal{P}_{\mathfrak{i}}(z)$ and it is purely arithmetical.
The process $\mathcal{D}_{\mathfrak{u}}((z)_{0},z_{1,0})$ is an instantiation of the parametric definition $\mathcal{D}_{\mathfrak{u}}(x,y)$
given by the following equation:
\begin{eqnarray*}
\mathcal{D}_{\mathfrak{u}}(x,y) &=& !e(v).\textit{if}\;1\le (v)_{0}\le x\;\textit{then}\;\textit{let}\;w=(v)_{0}-1\;\\
 && \ \textit{in}\;\textit{let}\;u=y_{w,0}\;\textit{in}\;\mathcal{S}_{\mathfrak{u}}([v_{1,1,u-1}/y_{w,1,u-1}]\ldots[v_{1,1,0}/y_{w,1,0}]y_{w,1,u}).
\end{eqnarray*}
The first parameter of $\mathcal{D}_{\mathfrak{u}}(x,y)$ indicates the number of the mutually dependent equations.
The second parameter is the G\"{o}del index of these parametric definitions that takes the form of (\ref{2011-11-05}).
The definition $\mathcal{D}_{\mathfrak{u}}(x,y)$ is essentially the interpretation of $\llbracket T\rrbracket_{\mathfrak{d}}$.
It is able to simulate all the parametric definitions that are encoded in $y$.
Again this is possible because the simulation is on-the-fly.
The simulator $\mathcal{S}_{\mathfrak{u}}(z)$ in (\ref{2011-07-09}) is defined in Fig.~\ref{Uvpc}.
The recursive functions $\textsf{r}_{6},\textsf{d}_{6}$ are similar to those of $\textsf{r}_{7},\textsf{d}_{7}$ respectively.
In case $\textsf{r}_{6}(z)=6$ the subroutine $\mathcal{D}_{\mathfrak{u}}((z)_{0},z_{1,0})$ is invoked with the parameter $\textsf{d}_{6}(z)$.

\begin{figure}[t]
\begin{center}
\begin{tabular}{l}
 ${\bf case}\ z\ {\bf of}$ \\
 $\ \ \textsf{r}_{6}(z){=}0 \;\Rightarrow\; {\bf 0}$; \\
 $\ \ \textsf{r}_{6}(z){=}1 \;\Rightarrow\; Nth(\textsf{d}_{6}(z)_{0},y).a_{y}(x).\mathcal{S}_{\mathfrak{u}}([x/\textsf{d}_{6}(z)_{1}]\textsf{d}_{6}(z)_{2})$; \\
 $\ \ \textsf{r}_{6}(z){=}2 \;\Rightarrow\; Nth(\textsf{d}_{6}(z)_{0},y).\overline{a_{y}}(val(\textsf{d}_{6}(z)_{1})).\mathcal{S}_{\mathfrak{u}}(\textsf{d}_{6}(z)_{2})$; \\
 $\ \ \textsf{r}_{6}(z){=}3 \;\Rightarrow\; \mathcal{S}_{\mathfrak{u}}(\textsf{d}_{6}(z)_{0})\,|\,\mathcal{S}_{\mathfrak{u}}(\textsf{d}_{6}(z)_{1})$; \\
 $\ \ \textsf{r}_{6}(z){=}4 \;\Rightarrow\; Nth(\textsf{d}_{6}(z)_{0},y).(a_{y})\mathcal{S}_{\mathfrak{u}}(\textsf{d}_{6}(z)_{1})$; \\
 $\ \ \textsf{r}_{6}(z){=}5 \;\Rightarrow\; \textit{if}\;val(\textsf{d}_{6}(z)_{0})\;\textit{then}\;\mathcal{S}_{\mathfrak{u}}(\textsf{d}_{6}(z)_{1})$; \\
 $\ \ \textsf{r}_{6}(z){=}6 \;\Rightarrow\; \overline{e}(\textsf{d}_{6}(z))$ \\
 ${\bf end}\ {\bf case}$ \\
\end{tabular}
  \caption{Simulator $\mathcal{S}_{\mathfrak{u}}(z)$.}
  \label{Uvpc}
\end{center}
\end{figure}

By an argument similar to the one given in the proof of Theorem~\ref{VPCeinVPC}, one can convince oneself of the validity of the following result.

\begin{thm}\label{2011-07-10}
$\mathbb{VPC}\in\mathbb{VPC}$.
\end{thm}

\section{Application}\label{sec-Application}

The existence of a universal process can be seen as an expressiveness criterion.
There cannot be any universal process for CCS~\cite{Milner1989} or the pure process-passing calculus~\cite{Sangiorgi1992,Thomsen1995} since neither is complete~\cite{FuYuxi}.
But once an interaction model does admit some sort of universal process, a whole range of new applications are available.
In this section we sketch three of them.

\subsection{Process Passing as Value Passing}\label{sec-Process-Passing-as-Value-Passing}

The most valuable contribution of a universal process is that it allows a receiving process to interpret a number as a process.
This is a very useful property when applying $\mathbb{VPC}$ to tackle practical programming issues.
But is it necessary to pass a process from one location to another?
If the process that appears in the target environment is completely the same as the one sent from the source environment, a positive answer to the question can hardly be convincing.
What is useful in practice is that the source process sends an {\em abstraction} parameterized over names, and the process on the receiving end instantiates the parameters of the abstraction by its local names.
This way of introducing the higher order feature in the $\pi$-calculus is adopted in~\cite{Sangiorgi1992,Sangiorgi1993}.
We follow the same approach to extend $\mathbb{VPC}$.
Our higher order $\mathbb{VPC}$ has the following grammar:
\[T := \ldots \mid X(a_{1},\ldots,a_{j}) \mid A(a_{1},\ldots,a_{j}) \mid a(X{:}\langle i,j\rangle).T \mid \overline{a}(A{:}\langle i,j\rangle).T,\]
where only the higher order terms are indicated.
An abstraction $A$ is a term whose global names are abstracted away.
A process with $j$ global names can be turned into an abstraction with $j$ parameters.
If for example $T$ is a term with the global names $c_{1},\ldots,c_{j}$ then $\lambda c_{1}\ldots c_{j}.T$ is an abstraction.
In the above syntax $A:\langle i,j\rangle$ indicates that $A$ is an abstraction with $i$ local names and $j$ parameters.
The term $a(X{:}\langle i,j\rangle).T$ is a higher order input, in which $\langle i,j\rangle$ provides the typing information of the abstraction variable $X$,
$\overline{a}(A{:}\langle i,j\rangle).T$ is in higher order output form,
$X(a_{1},\ldots,a_{j})$ is an instantiation of the abstraction variable $X$ at $a_{1},\ldots,a_{j}$ and $A(a_{1},\ldots,a_{j})$ an instantiation of the abstraction $A$ at $a_{1},\ldots,a_{j}$.
The instantiation of the abstraction $\lambda c_{1}\ldots c_{j}.T$ at $a_{1},\ldots,a_{j}$ is syntactically the same as the term $T\{a_{1}/c_{1}.\ldots,a_{j}/c_{j}\}$.
Instantiations must be type correct.
Formally the names in the higher order $\mathbb{VPC}$ are typed in the same way as the channels in the higher order $\pi$-calculus are typed~\cite{Sangiorgi1992,Sangiorgi1993}.
We ignore the type system in the present light weight treatment.
In addition to the operational semantics of $\mathbb{VPC}$ the higher order $\mathbb{VPC}$ has the following semantic rules:
\[\begin{array}{ccc}
\inference{}{a(X{:}\langle i,j\rangle).T\stackrel{a(A)}{\longrightarrow}T\{A/X\}}
\ & \inference{}{\overline{a}(A{:}\langle i,j\rangle).T\stackrel{\overline{a}(A)}{\longrightarrow}T}
\ & \inference{S\stackrel{a(A)}{\longrightarrow}S'\ \ \
T\stackrel{\overline{a}(A)}{\longrightarrow}T'}{S\,|\,T\stackrel{\tau}{\longrightarrow}S'\,|\,T'}
\end{array}\]

\vspace*{1mm}

\noindent In the higher order input rule $A$ must be of the same type as the variable $X$.

We can now explain how to simulate the operational semantics of the higher order $\mathbb{VPC}$ in the first order $\mathbb{VPC}$.
Let $\upsilon$ stand for a partial function from the set of abstraction variables to $\textsf{V}$ that is injective on its finite domain of definition.
The notation $\upsilon[Z{\rightarrow}z]$ refers to the function that is the same as $\upsilon$ except that it sends $Z$ onto $z$.
The nontrivial part of the encoding is given below:
\begin{eqnarray*}
\llbracket X(a_{1},\ldots,a_{j})\rrbracket_{\upsilon} &=& (d)(\overline{d}([\varsigma(a_{1})/v_1,\ldots,\varsigma(a_{j})/v_j]\upsilon(X))\,|\,\mathcal{U}_{d}^{i,a_{1}\ldots a_{j}}), \\
\llbracket a(X{:}\langle i,j\rangle).T\rrbracket_{\upsilon} &=& a(x).\llbracket T\rrbracket_{\upsilon[X{\rightarrow}x]},\ \mathrm{where}\ x\ \mathrm{is}\ \mathrm{chosen}\ \mathrm{such}\ \mathrm{that}\ \mathrm{it}\ \mathrm{is}\ \mathrm{not}\ \mathrm{in}\ T, \\
\llbracket \overline{a}(A{:}\langle i,j\rangle).T\rrbracket_{\upsilon} &=& \overline{a}(\llbracket A\rrbracket_{\upsilon}).\llbracket T\rrbracket_{\upsilon}.
\end{eqnarray*}
The operation $[\varsigma(a_{1})/v_1,\ldots,\varsigma(a_{j})/v_j](\_)$ replaces in $\upsilon(X)$ the indexes $v_1,\ldots,v_j$ of $j$ global names by the indexes of $a_{1},\ldots,a_{j}$ respectively.
Notice that since we are dealing with processes containing neither first order variables nor higher order variables, by the time the on-the-fly interpretation reaches $\upsilon(X)$, it has already been instantiated by a number.
Thus $v_1,\ldots,v_j$ can be safely calculated from $\upsilon(X)$.
The encoding of a higher order $\mathbb{VPC}$ process is given by $\llbracket P\rrbracket_{\emptyset}$, where $\emptyset$ is the nowhere defined function.

The translation of the higher order $\mathbb{VPC}$ into the first order $\mathbb{VPC}$ is notably different from the translation of the higher $\pi$ into the first order $\pi$~\cite{Sangiorgi1993}.
For one thing the simulation of the higher order communication in $\mathbb{VPC}$ is achieved by code transmission, whereas in the $\pi$ scenario the simulation is implemented via access control.
An advantage of our encoding is that it allows the `user' to exert tight control over the executions.
The user may wish to terminate the simulation after a certain amount of time, bounded by a time complexity function.
This can be implemented by incorporating into the encoding a timer that counts the number of simulation steps.

In practice higher order features are implemented in the same fashion.
In real world what is really passed over in a higher order communication is some code of a routine/process.
The code is not necessarily a number.
But the receiver interprets the code in the same way a universal process of $\mathbb{VPC}$ interprets a number.
See Section~\ref{sec-Programming-Paradigm} for more discussion.

\subsection{Communication Security}\label{sec-Communication-Security}

The completeness of an interaction model means that any encryption/decryption algorithm can be implemented in it and an encrypted text can be passed from one process to another.
By exploiting that fact, a universal process can offer an effective way to enhance the security of communications.
Suppose party A intends to send a piece of programme to party B through a public channel.
There is no way to prevent anyone from eavesdropping on the communication channels.
What party A can do is to encrypt the G\"{o}del number of the programme before sending it to party B.
After receiving the number, party B decodes the number to recover the G\"{o}del index.
It then places the encoded programme in its private environment and puts it into action by invoking a universal process.
Far more complicated scenarios can be designed along this line of thinking.

The point is that the existence of a universal process allows one to implement the well known security protocols in $\mathbb{VPC}$ to enhance communication security.
This is a more traditional approach compared to the one that introduces explicit operators to model security protocols in process calculi~\cite{AbadiGordon1999}.
In security world one assumes that all channels are insecure.
No matter what kind of operators are introduced in a calculus to enhance security of communication, they must be implemented using cryptographic techniques.
In one way or another one resorts to a universal process.

\subsection{Programming Paradigm}\label{sec-Programming-Paradigm}

If $\mathbb{VPC}$ is seen as a machine model, what would be an implementation of a higher order programming language in $\mathbb{VPC}$?
If $\mathbb{VPC}$ is seen as a programming language, what kind of programming paradigm does it support?
The issue of constructing a higher order programming language on top of a process model has been studied by several research groups, see the references in~\cite{SewellWojciechowskiUnyapoth2010}.
This is not the right place to overview the existing approaches and tools.
What we are going to do is to propose a new programming paradigm that makes essential use of universal processes.
To explain our idea we introduce a new process model referred to as $\mathbb{PL}$.
This is essentially the same as $\mathbb{VPC}$.
But instead of sending and receiving numbers, a $\mathbb{PL}$ process sends and receives strings.
The reason for a string-passing process calculus is that it provides the right level of abstraction to study programming theory in process algebra.
A piece of program is nothing but a string, which can be parsed, type checked and executed.

Let $\Sigma$ be a finite set of {\em symbols} and $\Sigma^{*}$ be the set of finite {\em strings} over $\Sigma$.
We will write $\varrho$ for a string variable and $\ell$ for a {\em string term} constructed from strings, string variables and string operations.
We write $\psi$ for a boolean formula about strings.
For simplicity we see a symbol as a string of length one.
We choose a set of basic string operations on $\Sigma^{*}$ and a set of basic logic operations on $\Sigma^{*}$.
We obviously need the head, tail and length operations, as well as the binary prefix relation among others.
The particular choice of the operations is not our concern.

The set of the $\mathbb{PL}$-terms is defined by the following BNF:
\begin{eqnarray*}
T &:=& {\bf 0} \mid a(\varrho).T \mid \overline{a}(\ell).T \mid T\,|\,T' \mid (c)T \mid \textit{if}\;\psi\;\textit{then}\;T \mid D(\ell_{1},\ldots,\ell_{k}).
\end{eqnarray*}
The operational semantics of $\mathbb{PL}$ is obtained from the labeled transition system of $\mathbb{VPC}$ by substituting strings for numbers.
The model $\mathbb{PL}$ is easily seen to be complete.
In fact there is clearly an effective bijection between $\Sigma^{*}$ and $\textsf{N}$.
Using this bijection it ought to be easy to construct two subbisimilarities to support the claim that $\mathbb{PL}\sqsubseteq\mathbb{VPC}$ and $\mathbb{VPC}\sqsubseteq\mathbb{PL}$.
We assume that $\Sigma$ contains the proper subset $\Sigma_{\textsf{N}}=\{0,\textsf{succ},+,\times\}$, which makes the definition of the length function possible.
The calculus $\mathbb{PL}_{\textsf{N}}$ defined in terms of the symbol set $\Sigma_{\textsf{N}}$ is as expressive as $\mathbb{VPC}$.
We may think of $\mathbb{PL}_{\textsf{N}}$ as a machine model on which $\mathbb{PL}$ is implemented, by which we mean that there is an interpreter of $\mathbb{PL}$ in $\mathbb{PL}_{\textsf{N}}$.

Now the general framework has been set up, let's explain how programming languages can be implemented in our model.
Suppose $\mathfrak{L}$ is a concurrent, typed programming language, which could be as sophisticated as a full-fledged programming language or as simple as the concurrent language studied in~\cite{Milner1989}.
The {\em definition} of $\mathfrak{L}$ is given by a parser $\{\mathfrak{P}_{c}\}_{c\in\mathcal{N}}$.
Let $Pr$ be an $\mathfrak{L}$ program.
Then
\[\overline{c}(Pr)\,|\,\mathfrak{P}_{c}\stackrel{\iota}{\longrightarrow}L\]
for some $L$.
The process indicates acceptance if $Pr$ is a well defined $\mathfrak{L}$ program,
it refuses if $Pr$ violates the $\mathfrak{L}$ grammar.
The {\em implementation} of $\mathfrak{L}$ is given by an executor $\{\mathfrak{E}_{c}^{i,\jmath}\}_{i\in\textsf{N},\jmath\in\mathcal{N}^{*},c\notin\jmath}$, which is feasible by the technique developed in this paper.
For a legitimate $\mathfrak{L}$ program $Pr$ of the right type one must have that
\[\overline{c}(Pr)\,|\,\mathfrak{E}_{c}^{i,\jmath}\stackrel{\iota}{\longrightarrow}I\]
for some $I$ that implements $Pr$ in $\mathbb{PL}$.
This oversimplified account should be enough to give the reader a taste of our methodology.

Different programming paradigms are supported by different process models.
If one intends to study the object oriented features, one uses the $\pi$-calculus.
Since $\pi$ is also complete~\cite{FuYuxi}, everything carried out in this paper can be repeated for $\pi$.
It should not be difficult to internalize, as it were, Walker's meta translation of an object oriented language~\cite{Walker1991,Walker1995} in the fashion advocated here.
What we get is an implementation, rather than a translation, of the object oriented language in $\pi$.

\section{Recursion Theory in $\mathbb{VPC}$}\label{sec-S-m-n-Theorem}

In this section we explain how to do recursion theory in $\mathbb{VPC}$ by taking a look at $\mathbb{VPC}$ version of the S-m-n Theorem.
The challenge here is actually how to formulate it correctly.
We know from the recursion theory that S-m-n Theorem is about partial evaluation.
There is an effective way to transfer the index of a $(k_{0}{+}k_{1})$-ary effective function $\textsf{f}(x_{1},\ldots,x_{k_{0}},y_{1},\ldots,y_{k_{1}})$ and the inputs $i_{1},\ldots,i_{k_{0}}$ to the index of the $k_{1}$-ary effective function $\textsf{f}(i_{1},\ldots,i_{k_{0}},y_{1},\ldots,y_{k_{1}})$.
If we are ever to have a recursion theory of $\mathbb{VPC}$ processes, we must start by answering the question of what the right $\mathbb{VPC}$ counterpart of a recursive function is.
It is not hard to see that the most natural generalization of a recursive function is a parametric definition of the form
\begin{equation}\label{2011-08-31}
D(z_{1},\ldots,z_{k}) = T.
\end{equation}
For numbers $i_{1},\ldots,i_{j}$, where $j\le k$, we will write $D(i_{1},\ldots,i_{j},z_{j+1},\ldots,z_{k})$ for the $(k{-}j)$-ary parametric definition $D'(z_{j+1},\ldots,z_{k})$ given by
\[D'(z_{j+1},\ldots,z_{k}) = T\{i_{1}/z_{1},\ldots,i_{j}/z_{j}\}.\]
Suppose $i$ is the number of local names of $T$ and $\jmath$ is the list of the global names in $T$.
We say that the $D(z_{1},\ldots,z_{k})$ defined in (\ref{2011-08-31}) is a $k$-ary parametric definition of type $[i,\jmath]$.
Two $k$-ary parametric definitions, say $D_{0}(z_{1},\ldots,z_{k})$ and $D_{1}(z_{1},\ldots,z_{k})$, are equal if for all numbers $i_{1},\ldots,i_{k}$, one has $D_{0}(i_{1},\ldots,i_{k})=_{\mathbb{VPC}}D_{1}(i_{1},\ldots,i_{k})$.

By recycling the encoding in Section~\ref{sec-Godel-Index} we can G\"{o}delize the set of the $k$-ary parametric definitions of type $[i,\jmath]$.
Our technique to derive a universal process, as developed in Section~\ref{sec-Universal-Process}, helps define a universal process $\mathcal{U}^{[i,\jmath][k]}_{z}(z_{1},\ldots,z_{k})$ for the set of the $k$-ary parametric definitions of type $[i,\jmath]$.
Here the word `process' is not very precise since $\mathcal{U}^{[i,\jmath][k]}_{z}(z_{1},\ldots,z_{k})$ is a parametric definition rather than a process.
The parameter $z$ is made an index, suggesting that $\mathcal{U}^{[i,\jmath][k]}_{z}(z_{1},\ldots,z_{k})$ should be thought of as the $z$-th $k$-ary parametric definition of type $[i,\jmath]$.
The defining property of $\mathcal{U}^{[i,\jmath][k]}_{z}(z_{1},\ldots,z_{k})$ requires that for each number $j$, say the index of $D(z_{1},\ldots,z_{k})$, the equality \[\mathcal{U}^{[i,\jmath][k]}_{j}(z_{1},\ldots,z_{k}) =_{\mathbb{VPC}} D(z_{1},\ldots,z_{k})\] holds.
Now we can state the S-m-n Theorem.

\begin{thm}\label{s-m-n}
Suppose $k_{0},k_{1}$ are natural numbers.
There is a total $(k_{0}{+}1)$-ary recursive function $\textsf{s}_{k_{1}}^{k_{0}}(z,x_{1},\ldots,x_{k_{0}})$ such that for all numbers $j,i_{1},\ldots,i_{k_{0}}$ the following equality holds:
\[\mathcal{U}^{[i,\jmath][k_{0}+k_{1}]}_{j}(i_{1},\ldots,i_{k_{0}},y_{1},\ldots,y_{k_{1}}) =_{\mathbb{VPC}} \mathcal{U}^{[i,\jmath][k_{1}]}_{\textsf{s}_{k_{1}}^{k_{0}}(j,i_{1},\ldots,i_{k_{0}})}(y_{1},\ldots,y_{k_{1}}).\]
\end{thm}
\begin{proof}
The proof follows the standard argument.
Given the index of a $(k_{0}{+}k_{1})$-ary parametric definition $D''(x_{1},\ldots,x_{k_{0}},y_{1},\ldots,y_{k_{1}})$ of type $[i,\jmath]$, one can effectively construct the $k_{1}$-ary parametric definition $D''(i_{1},\ldots,i_{k_{0}},y_{1},\ldots,y_{k_{1}})$ of the same type, from which we get its G\"{o}del index effectively.
This defines a total recursive function.
\end{proof}

The S-m-n Theorem helps import results in recursion theory to process theory.
The famous Rice Theorem~\cite{Rogers1987} is one such result.

\begin{thm}\label{rice}
Suppose $\mathcal{B}$ is a set of $k$-ary parametric definitions that satisfies the following:
\begin{enumerate}
\item \label{2014-09-03} $\mathcal{B}$ is not empty;
\item there is some $k$-ary parametric definition that is not in $\mathcal{B}$;
\item $\mathcal{B}$ is closed under the absolute equality.
\end{enumerate}
Then the set $\{j \mid \mathcal{U}^{[i,\jmath][k]}_{j}(x_{1},\ldots,x_{k})\in\mathcal{B}\ \mathrm{for}\ \mathrm{some}\ [i,\jmath]\}$ is undecidable.
\end{thm}
\begin{proof}
Without loss of generality, suppose $\Omega\notin\mathcal{B}$.
By condition (\ref{2014-09-03}) the set $\mathcal{B}$ contains some $k$-ary parametric definition $D(x_{1},\ldots,x_{k})$.
Let $W$ contain the number $\langle i_{1},\ldots,i_{k}\rangle$ if the $\langle i_{1},\ldots,i_{k}\rangle$-th $k$-ary recursive function is definable at $i_{1},\ldots,i_{k}$.
It is well known that $W$ is recursive enumerable but not decidable.
Let the $(k{+}1)$-ary parametric definition $D'(z,x_{1},\ldots,x_{k})$ be defined by
\begin{eqnarray*}
D'(z,x_{1},\ldots,x_{k}) &=& \textit{if}\;z\in W\;\textit{then}\;D(x_{1},\ldots,x_{k}).
\end{eqnarray*}
Suppose $D'(z,x_{1},\ldots,x_{k})$ is of type $[i,\jmath]$.
Then some number $j$ exists such that
\[\mathcal{U}^{[i,\jmath][k+1]}_{j}(z,x_{1},\ldots,x_{k}) =_{\mathbb{VPC}} D'(z,x_{1},\ldots,x_{k}).\]
According to Theorem~\ref{s-m-n} some binary total recursive function $\textsf{s}$ exists such that
\[\mathcal{U}^{[i,\jmath][k]}_{\textsf{s}(j,z)}(x_{1},\ldots,x_{k}) =_{\mathbb{VPC}} \mathcal{U}^{[i,\jmath][k+1]}_{j}(z,x_{1},\ldots,x_{k}).\]
Using (3) it is clear that $k\in W$ if and only if $D'(k,x_{1},\ldots,x_{k})=_{\mathbb{VPC}}D(x_{1},\ldots,x_{k})$ if and only if $\mathcal{U}^{[i,\jmath][k]}_{\textsf{s}(j,k)}(x_{1},\ldots,x_{k})\in\mathcal{B}$.
So $\mathcal{B}$ cannot be decidable.
\end{proof}

There is nothing new about the above proof.
But at least it demonstrates that the type constraint $[i,\jmath]$ of $\mathcal{U}^{[i,\jmath][k]}$ is not much of a restriction.

A simple consequence of the Rice Theorem is about the unobservable processes.

\begin{cor}
The set of the unobservable processes is undecidable.
\end{cor}

\section{Future Work}\label{sec-Future-Work}

The idea of designing universal processes was discussed in~\cite{AndersenMorkSorensen1997} in the framework of CCS.
Due to the limitation of the model, a universal process for CCS is static in the sense that it must preload the G\"{o}del number of the CCS process to be simulated since it can never dynamically input any number.
In order to simulate the branching structure of a process, the universal process for CCS must introduce divergence.
A recent work is reported in~\cite{BaetenLuttikTilburg2011}, where the authors studied universal Reactive Turing Machines.
A universal Reactive Turing Machine either introduces divergence, or places restriction on the maximum branching degree of the Reactive Turing Machine being simulated, both divergence and branching degree being of semantic nature.
A universal machine for Reactive Turing Machines is also static in the above sense since the transmission of the description of a machine, which is a string of symbols, to the universal machine should not be interrupted.
The advantage of our universal process is that it is dynamic and does not impose any {\em semantic} constraints on any processes to be simulated.

We have chosen to study universal processes of $\mathbb{VPC}$ since natural numbers are familiar and there is a handy decidable Presburger theory.
The model introduced in Section~\ref{sec-Programming-Paradigm} is built on another value domain.
To achieve completeness very little is required from value domain.
Completeness has more to do with the communication mechanism of a model than with the value domain of the model.
The constructions of the universal processes in other complete models might look quite different.
In $\pi$-calculus for example the numbers can be coded up using the following inductively defined name indexed functions:
\begin{eqnarray*}
\llbracket 0\rrbracket_{a} &=& \overline{a}(c).\overline{c}(b).\overline{c}(e).\overline{e}e, \\
\llbracket i{+}1\rrbracket_{a} &=& \overline{a}(c).\overline{c}(b).\overline{c}(e).\llbracket i\rrbracket_{b}.
\end{eqnarray*}
It is more or less a formality to apply our approach to generate a universal process of $\pi$.

For a complete model the existence of universal process depends on the operators of the model.
Since interpretation is done on-the-fly, problem may arise if an operator is not congruent over the absolute equality.
The best counter example is given by the unguarded choice.
Upon receiving the G\"{o}del number $k$ for $P'+P''$, an interpreter $\mathcal{I}$ is expected to generate $\mathcal{I}(k')+\mathcal{I}(k'')$, where $k'$ codes up $P'$ and $k''$ codes up $P''$.
But this would be disastrous because $\mathcal{I}(k')$ would normally engage in some internal actions before the real simulation of $P$ happens.
In other words, it can preempt the choice.
Guarded choice would not have this problem.

Is there a reasonable (complete) model that does not have any universal process?
How strong is it to formulate a thesis postulating that all complete models have universal processes?
This is an important issue to be studied.

The idea of universal process can be further exploited.
Two important directions are outlined below.
\begin{enumerate}
\item
At the theoretical level, it is interesting to see how recursion theory~\cite{Rogers1987,Soare1987} defined in $\mathbb{VPC}$ can help develop an interactability theory of the model.
Interactability theory aims to study definability of interactive behaviours, rather than computational behaviours, in interactive models.
\item
At the programming level, it is worth the effort to study programming theory in a systematic way.
More generally we can look at the class of process calculi with universal processes.
These are models well equipped to model programming features.
The significance of these models to programming theory is yet to be investigated.
\end{enumerate}

\noindent Is it possible for a universal process to be a single process rather than a family of processes?
A drastic approach to address the issue is to introduce a bijective naming function $\nu_{(\_)}:\textsf{N}\rightarrow\mathcal{N}$ that a $\mathbb{VPC}$ process may make use of. The function $\nu_{(\_)}$ gives an enumeration $\nu_{0},\nu_{1},\ldots$ of the names.
It can be extended to a function from $\textsf{T}$ to the set $\{\nu_{t} \mid t\in\textsf{T}\}$ of name expressions.
Now the grammar of $\mathbb{VPC}_{\nu}$ can be defined by the following BNF:
\begin{eqnarray*}
T &:=& {\bf 0} \mid \nu_{t}(x).T \mid \overline{\nu_{t}}(t).T \mid T\,|\,T \mid (\nu_{i})T \mid \textit{if}\;\psi\;\textit{then}\;T \mid D(t_{1},\ldots,t_{k}).
\end{eqnarray*}
An example of a $\mathbb{VPC}_{\nu}$ process is
\[\nu_{0}(x).\nu_{1}(y).\textit{if}\;x=2y\;\textit{then}\;\overline{\nu_{x}}(y).\]
It is clear from this example that $\mathbb{VPC}_{\nu}$ admits mobile computing to a certain degree.
Notice that the match operator $[\nu_{t}{=}\nu_{t'}]T$ is definable in $\mathbb{VPC}_{\nu}$.
The new model lacks the power, and the trouble as well, introduced by relocating local names.
The virtue of $\mathbb{VPC}_{\nu}$ is that it has a single universal process $\mathcal{U}_{c}$ that is capable of dealing with indices of all $\mathbb{VPC}_{\nu}$ processes.
This is rendered possible by the naming function which produces a canonical indexing for all the names whatsoever.
Further study on $\mathbb{VPC}_{\nu}$ is necessary before we can say more about its theoretical and practical relevance to process theory.

\section*{Acknowledgement}

This work has been supported by NSFC (61472239, PACE 61261130589).
Xiaojuan Cai's idea about the on-the-fly simulations has been very instructive to this work.
Sandy Harris and Huan Long have helped in improving the quality of this paper.


\end{document}